\documentclass[article]{siamart0516}

\usepackage{verbatim,amsmath,amsfonts,epsfig,graphics,setspace,amssymb,
float,stmaryrd, multirow}
\usepackage{caption}
\usepackage{epsfig}
\usepackage{enumerate,color}
\usepackage{url}
\usepackage{algorithm}
\usepackage{algorithmic}

\newcounter{rmrk}[section]

\newcounter{example}[section]

\newtheorem{assumption}{Assumption}[section]
\newcommand{\norm}[1]{\left\Vert#1\right\Vert}

\newcommand{\set}[1]{\left\{#1\right\}}

\newcommand{\mbf}[1]{\mathbf{#1}}

\newcommand{\pspace}{{\Lambda}}
\newcommand{\dspace}{{\mathcal{D}}}
\newcommand{\dspacealpha}{{\mathcal{D}_\alpha}}
\newcommand{\Qspace}{{\mathcal{Q}}}

\newcommand{\pmeas}{\mu_{\pspace}}
\newcommand{\dmeas}{\mu_{\dspace}}
\newcommand{\dmeasalpha}{\mu_{\dspace_\alpha}}
\newcommand{\pborel}{\mathcal{B}_{\pspace}}
\newcommand{\dborel}{\mathcal{B}_{\dspace}}
\newcommand{\dborelalpha}{\mathcal{B}_{\dspace_\alpha}}

\newcommand{\initmeas}{\mathbb{P}_{\pspace}^{\text{init}}}

\newcommand{\initdens}{\pi_{\pspace}^{\text{init}}}
\newcommand{\updens}{\pi_{\pspace}^{\text{up}}}
\newcommand{\predmeas}{\mathbb{P}_{\dspace}^{\text{pred}}}

\newcommand{\preddensalpha}{\pi_{\dspace_\alpha}^{\text{pred}}}

\newcommand{\obsmeasalpha}{\mathbb{P}_{\dspace_\alpha}^{\text{obs}}}
\newcommand{\obsdensalpha}{\pi_{\dspace_\alpha}^{\text{obs}}}

\newcommand{\harmmeas}{\mu_{\pspace}^{\text{hm}}}

\newcommand{\param}{\lambda}
\newcommand{\qmap}{Q}

\title{Optimal Experimental Design Criteria for Data-Consistent Inversion}
\date{\today}

\author{
  Troy Butler\thanks{Department of Mathematical and Statistical Sciences, University of Colorado Denver, Denver, CO 80202 ({\tt Troy.Butler@ucdenver.edu})}
  \and
    John  Jakeman\thanks{Optimization and Uncertainty Quantification Department, Center for Computing Research, Sandia National Labs, Albuquerque, NM 87185.}
  \and
  Michael Pilosov\thanks{Mind the Math, LLC, Denver, CO 80203.}
    \and
  Scott Walsh\thanks{QuantumScape, Denver, CO.}
  \and
    Timothy Wildey\thanks{Computational Mathematics Department, Center for Computing Research, Sandia National Labs, Albuquerque, NM 87185.}
}

\begin{document}
\maketitle

\begin{abstract}

The ability to design effective experiments is crucial for obtaining data that can substantially reduce the uncertainty in the predictions made using computational models. 
An optimal experimental design (OED) refers to the choice of a particular experiment that optimizes a particular design criteria, e.g., maximizing a utility function, which measures the information content of the data. 
However, traditional approaches for optimal experimental design typically require solving a large number of computationally intensive inverse problems to find the data that maximizes the utility function.
Here, we introduce two novel OED criteria that are specifically crafted for the data consistent inversion (DCI) framework, but do not require solving inverse problems. 
DCI is a specific approach for solving a class of stochastic inverse problems by constructing a pullback measure on uncertain parameters from an observed probability measure on the outputs of a quantity of interest (QoI) map. 
While expected information gain (EIG) has been used for both DCI and Bayesian based OED, the characteristics and properties of DCI solutions differ from those of solutions to Bayesian inverse problems which should be reflected in the OED criteria. 
The new design criteria developed in this study, called the expected scaling effect and the expected skewness effect, leverage the geometric structure of pre-images associated with observable data sets, allowing for an intuitive and computationally efficient approach to OED. 
These criteria utilize singular value computations derived from sampled and approximated Jacobians of the experimental designs. 
We present both simultaneous and sequential (greedy) formulations of OED based on these innovative criteria. 
Numerical results demonstrate the effectiveness in our approach for solving stochastic inverse problems. 
\end{abstract}

\begin{keywords}
optimal experimental design, singular values, uncertainty quantification, inverse problems
\end{keywords}


\section{Introduction}
\label{sec:Intro}




Quantification of model uncertainties, especially as they pertain to model predictions, is essential for credible simulation-aided prediction and discovery. 
The uncertainty in model predictions is often sensitive to the characterization of the input uncertainties. 
To characterize uncertainties on model inputs, one often specifies a set of physically plausible ranges or assigns a probability distribution by applying either engineering or domain-specific knowledge. 
Unfortunately, these prior descriptions of uncertainty are typically not validated by data and can produce overly conservative or unreliable estimates of uncertainty on model predictions. 
Consequently, when possible, it is advantageous to solve an inverse problem to condition prior estimates of uncertainty on observational data. 
By carefully selecting the data, solving these inverse problems can significantly reduce model uncertainty.
However, the collection of observational data on quantities of interest (QoI) is often an expensive task.
Furthermore, it is often unclear a priori how large of an impact such QoI data will have on reducing uncertainties on input parameters, and subsequently on any model predictions.
In this paper, we propose new experimental design criteria to determine an optimal experimental design (OED) based on the geometric properties of the associated QoI map induced by the model between parameters and the QoI data.


Simple experimental design approaches, such as factorial or Latin-Hypercube design, produce experiments that are well spaced in the design space defined by the potential QoI maps~\cite{Cox_R_Book_2000,Fisher_Book_1966,Pazman_book_1986}. 
Such approaches do not exploit any knowledge of the physical process they are measuring. 
Using modeling and simulation to guide experimental design can often produce much more effective design criteria~\cite{Korkel_KB_S_OMS_2004,Chung_H_SIAMCO_2012,Bock2013,Horesh2010,Bartuska_ET_CMAME_2022,Foster_JBHTRG_ANIPS_2019,Catanach_D_IEEE_2023}.
In this context, OED often refers to the process of using simulations to select which experiments, defined by their associated QoI maps, to perform in order to collect the associated QoI data.
This simulation-aided OED process generally seeks to optimize a particular metric (sometimes called a utility function), which encodes the intended purpose of the data to be collected. 

When the QoI depend linearly on the model inputs, many design criteria can be expressed as functions of the Fisher information matrix~\cite{Atkinson_book,Chaloner_V_IMS_1995,Kouri_JH_SIAMUQ_2021}. 
For example, A-optimality minimizes the trace of the inverse of the information matrix which minimizes the average variance of the input uncertainty. 
In addition, D-optimality maximizes the determinant of the information matrix, which maximizes the expected information gain (EIG). 
When the QoI depend nonlinearly on the inputs, minimizing criteria based on the information matrix is often insufficient. 
For nonlinear QoI maps, a decision-theoretic approach that maximizes expected utility is often adopted~\cite{Ryan_DMP_ISR_2016}. 
Given a particular QoI map, the standard Bayesian approach for finding the OED is to first approximate, or generate samples from, the posterior densities and use these densities to compute the expected utility of the observable map~\cite{Chaloner_V_IMS_1995, bernardo, vanderberg_bayesianoed}. 
A commonly used utility function in this setting is the EIG~\cite{bernardo}, which is derived from the Kullback-Leibler (KL) divergence~\cite{KL1951, NME5211,KL1951, RenyiKL2014, Huan_2014, WWJ17} and describes the difference between posterior and prior densities.
While the KL divergence is by no means the only way to compare two probability densities, it does provide a reasonable measure of the information gained in the sense of Shannon information \cite{Cover} and is commonly used in Bayesian OED \cite{Huan_2014, NME5211}.



OED procedures should be tailored to the model purpose and the inverse problem used to update model uncertainties. 
In this paper, we present OED strategies that can be used to judiciously choose which QoI to utilize for data consistent inversion (DCI). 
DCI was originally designed to solve a class of stochastic inverse problem that seeks a pullback probability measure to characterize aleatoric uncertainties rather than to inform epistemic uncertainties~\cite{DCI}. 
Existing work on developing OED strategies for DCI adapted ideas from Bayesian inference to utilize the KL divergence as the utility function to find the optimal design~\cite{WWJ17}.
In~\cite{butler_oed4p}, a similar approach was utilized to defined the OED by maximizing the KL divergence on predictions rather than parameters.
However, neither of these approaches for OED exploited the measure-theoretic properties that are foundational to the DCI framework.
Additionally, the particular class of inverse problems solved by DCI are fundamentally different from Bayesian and frequentist inverse problems and thus alternative OED strategies may be required.


The main focus of this work is to analyze the use of specific geometric properties, called skewness and scaling effects, of potential QoI maps as design criteria for OED problems.
These properties are rooted in measure theory, require very few assumptions on both the spaces of model inputs and QoI data (as well as the associated QoI maps), and are developed {\em independently} (and a priori) of any particular type of inverse problem framework.
While it is possible to use the OED obtained by this method to solve different types of inverse problems, such as Bayesian formulations,
this work focuses on the implications of this approach for DCI-based solutions.
This work includes several specific contributions to the field of OED:
\begin{itemize}
\item a rigorous theoretical framework for geometric-based OED;
\item a straightforward approach to evaluate the geometric OED criteria from singular values of sampled Jacobian matrices of the QoI maps;
\item demonstration of the utilization of these geometric OED criteria for simultaneous design where multiple experiments are chosen in one round of optimization;
\item a greedy algorithm for sequential design, which identifies experiments in multiple rounds of optimization based on identifying experiments that provide complementary information relative to the previously chosen experiments.
\end{itemize}
%

The rest of this paper is structured as follows.
Section~\ref{sec:DCI-review} reviews the background material on DCI and provides the notation for this work.
Section~\ref{sec:QoI_practical_restrictions} presents a theoretical argument for performing experiment selection based on the effective dimensional of high-dimensional maps.
Section~\ref{sec:scaling} summarizes the geometric criterion referred to as the {\em scaling effect} of maps, which is related to the measure of inverse events.
Section~\ref{sec:skewness} introduces a separate geometric criterion referred to as the {\em skewness effect} of a map, which is related to the accuracy of approximating inverse events.
This criterion is useful in distinguishing between maps with similar scaling effects.
With potentially different types of criteria to consider, we formally define the OED problem in Section~\ref{sec:oed}.
We then describe skewness as means of quantifying changes to scaling by the addition of new measurements to motivate the greedy algorithm in Section~\ref{subsec:greedy}.
This is followed by an application of both criteria to identify OEDs for the heating of a thin metal rod with uncertain thermal conductivities in Section~\ref{S:Example}.
Section~\ref{sec:greedy_results} demonstrates the greedy algorithm to design a 9-dimensional observable map that is both intuitive and interpretable in its construction.
Concluding remarks and possible future research directions are provided in Section~\ref{sec:Conclusions}.
Proofs of technical results are provided in Appendices~\ref{app:dimension_assump_theorem} and~\ref{app:scaling-skewness-proofs}.


\section{Background, Notation, and Terminology}\label{sec:DCI-review}


Computational models define a particular (vector-valued) map from model inputs to model outputs, which we either refer to as simply the quantities of interest (QoI) map, or just the map if the context is clear.
Let 
\[
\lambda=(\lambda_1,\ldots,\lambda_n)^\top\in\pspace\subset\mathbb{R}^n
\]
denote a sample of uncertain inputs of a simulation model that are contained within the parameter space denoted by $\pspace$.
For simplicity, we refer to all model inputs as either model parameters or just parameters when the context is clear.
Now, let 
\[
Q(\lambda)\in\dspace:=\set{Q(\lambda)\, : \, \lambda\in\pspace} \subset\mathbb{R}^m
\]
define a particular QoI map whose range, denoted by $\dspace$, is called the data space.
We frequently utilize the notation $\qmap^{-1}(A)$ to denote the preimage of the set $A$, i.e.,
\[
\qmap^{-1}(A) := \left\{ \param \in \pspace \ : \ \qmap(\param)\in A \subset \dspace\right\}.
\]
This notation is commonly used in measure theory and does not require the map to be invertible.

Before data are collected, there are numerous (potentially infinite) experimental designs defining which QoI data are to be collected with each corresponding to a different QoI map.
Each experiment consists of a set of observations that the model must be able to predict.
Let $\Qspace$ denote the space of all locally differentiable QoI maps defining the designs under consideration by the modeler. 
We refer to $\Qspace$ as the {\em design space}. 

When general properties of QoI maps are considered in this paper, we usually refer to an arbitrarily selected $Q\in\Qspace$.
However, when it is necessary to specify distinct QoI maps within $\Qspace$, we assume that an index set $\mathcal{A}$ exists such that $\Qspace = \set{Q_{\alpha}}_{\alpha\in\mathcal{A}}$.
Similarly, we use $\dspace_{\alpha}$ to denote the data space associated with the map $Q_{\alpha}$. 
When $Q_{\alpha}$ is vector-valued, we use $Q_{\alpha,i}$ to denote the $i$th component of $Q_{\alpha}$ and $\dspace_{\alpha,i}$ to denote the data subspace associated with this component map.
It is worth noting that correlations may exist between components of any given map, and therefore, we do not expect $\dspace_{\alpha}$ to be a Cartesian product of its associated data subspaces.

Throughout this paper, we use $(\pspace, \pborel, \pmeas)$ and $(\dspacealpha, \dborelalpha, \dmeasalpha)$ to denote the measure spaces, with Borel $\sigma$-algebras, that define the parameter and associated QoI data spaces, respectively.
We also assume that the measures $\pmeas$ and $\dmeasalpha$ are dominating measures for any probability measure used on these spaces, which are useful for defining the Radon-Nikodym derivatives (i.e., probability density functions) of probability measures. 
We use Lebesgue measures in this work; however, this is not a requirement.

While the two novel OED criteria developed in this paper are not limited to utilization with the DCI framework, we focus on demonstrating the impact of these criteria on the solution of DCI-based inverse problems. 
The basic DCI formulation seeks to constrain the aleatoric (irreducible) uncertainty of data on a population using the pullback of an observed probability measure, $\obsmeasalpha$, defined on $(\dspacealpha, \dborelalpha)$ using $\qmap_{\alpha}$. 
In other words, given an observed probability measure, $\obsmeasalpha$, on $(\dspacealpha, \dborelalpha)$, DCI seeks a probability measure, $\mathbb{P}_\pspace$, on $(\Lambda, \mathcal{B}_{\Lambda})$ such that the push-forward of $\mathbb{P}_\pspace$ through $\qmap_\alpha$ matches the observed measure, i.e.,
\begin{equation}
\obsmeasalpha(A) = \mathbb{P}_\pspace(\qmap_{\alpha}^{-1}(A)) \quad \forall A \in \dborelalpha.\label{eq:consistent}
\end{equation}


Following~\cite{DCI}, we construct a unique and stable DCI solution by first defining
an {\em initial} probability measure, $\initmeas$, on $(\Lambda, \mathcal{B}_{\Lambda})$, which we update to satisfy~\eqref{eq:consistent} by comparing its push-forward measure, defined as a prediction, denoted by $\predmeas$, to the observed measure.
As stated above, we assume that all probability measures are absolutely continuous with respect to the dominating measures so that they admit Radon-Nikodym derivatives (which are probability densities in this work).
As in \cite{DCI}, we guarantee the {\em existence} of a solution by making the following {\em predictability assumption} involving these densities:
\begin{assumption}[Predictability Assumption]\label{assump:dom}
\normalfont
For each $Q_\alpha\in\Qspace$, there exists a constant $C_\alpha>0$ such that
\[\obsdensalpha(q)\leq C_p\preddensalpha(q), \quad \text{for a.e.} \hspace{0.1cm} q\in\dspacealpha.\]
\end{assumption}

\Cref{assump:dom} implies that the observed probability measure is absolutely
continuous with respect to the predicted probability measure.
This assumption is also related to the ability to construct numerical approximations of the updated density defining the DCI solution since the constant $C_\alpha$ is the same constant utilized/estimated when performing rejection sampling.
Under the predictability assumption, a disintegration theorem is used to give a unique update to the initial probability measure that defines a DCI solution.
We refer to this probability measure as the {\em updated} measure on $(\Lambda, \mathcal{B}_{\Lambda})$, and its associated updated density is given by
\begin{eqnarray}\label{eq:upd_dens}
\updens(\param) = \initdens(\param)r_\alpha(\param), \quad \text{where} \quad r_\alpha(\param) =\frac{\obsdensalpha(\qmap_\alpha(\param))}{\preddensalpha(\qmap_\alpha(\param))}.
\end{eqnarray}

In practice, the ratio $r_\alpha(\param)$ updates the initial density to construct a solution to the stochastic inverse problem for a particular QoI map.  
The sample average of the ratio (computed from samples generated by the initial distribution) also provides a useful diagnostic for numerical validation of the predictability assumption since the expected value corresponds to integrating the updated density and should therefore be unity.
We refer the interested reader to~\cite{DCI}  for additional details, the full measure-theoretic derivation, and a comparison between the updated density and the standard Bayesian posterior.  

\section{Effective Dimensionality of QoI maps}\label{sec:QoI_practical_restrictions}

Conducting experiments (or collecting data from the field) is often an expensive task that can restrict data acquisition to a limited number of observable QoI, thereby limiting the dimension of the QoI map. 
In the simulation realm, researchers studying OED in a Bayesian context have documented that a limit on the number of useful dimensions in the QoI map becomes quickly evident and motivates the sparsification (i.e., dimension reduction) of experimental designs; e.g., see \cite{APS+14, HHT2010, Haber:2012}.
While these provide some motivation for considering lower-dimensional QoI maps in the OED process, we provide a measure-theoretic justification for an assumption on the limit on the number of practically useful dimensions of the QoI maps considered in the design spaces of this work.

For notational simplicity, we let $Q$ denote an arbitrarily chosen QoI map from the design space and let $J_{Q}(\lambda)$ denote the Jacobian of this map at a point $\lambda\in\pspace$. 
Let $r$ denote the rank of $J_{Q}(\lambda)$,  and assume, at least initially, that this rank is constant over $\pspace$.
Then, $r\leq \min\set{m,n}$, and there exists a piecewise-linear $r$-dimensional map $\widehat{Q}$ that approximates $Q^{-1}(E)$ for all events $E\in\dborel$ to arbitrary accuracy in $\pmeas$-measure as summarized in the following theorem.

\begin{theorem}\label{thm:local_inverse}
	Suppose $Q\in\Qspace$ is an $m$-dimensional map, $J_{Q}(\lambda)$ has constant rank $r$ in $\pspace$, $\pmeas$ is a product measure, and $E$ is any set in $\dborel$.
	For every $\epsilon>0$, there exists a piecewise-linear $r$-dimensional map, $\widehat{Q}:\pspace\to\widehat{\dspace}$, and a finite number $K$ of $r$-dimensional generalized rectangles, $\set{\widehat{E}_k}_{1\leq k\leq K}\subset\widehat{\dborel}$, such that if $\Delta$ denotes the symmetric difference of $Q^{-1}(E)$ and $\bigcup_{1\leq k\leq K}\widehat{Q}^{-1}(\widehat{E}_k)$, then $\pmeas(\Delta)<\epsilon$.   
\end{theorem}

\begin{proof}
    See Appendix~\ref{app:dimension_assump_theorem}.
\end{proof}

We focus on two important implications of this result in the context of this work. 
First, Theorem~\ref{thm:local_inverse} implies that it is theoretically possible to post-process any observable map $Q$ as an $r$-dimensional map $\widehat{Q}$ that maintains the accuracy of inverses of observable events.
Theorem~\ref{thm:local_inverse} also implies that we can replace any computations of the geometric effects defining utility functions on $Q$ with utility functions involving $\widehat{Q}$.
It follows that since $r$ is bounded above by the minimum of $n$ (i.e., $\dim(\pspace)$) and $m$ (i.e., $\dim(\dspace_Q)$) that we can assume with no loss of generality that $m\le n$.

It is worth noting that there are situations where an intermediary step exists such that the $m$ QoI are not specified but rather are learned from high-dimensional and noisy data sets, for instance, through the application of feature extraction techniques such as kernel Principal Component Analysis (kPCA), which is common in the machine learning and data science communities, e.g., see~\cite{pearson1901liii,scholkopf1997kernel,mika1999kernel,abdi2010principal}. 
The utilization of kPCA within the DCI framework to learn lower-dimensional QoI from noisy temporal or spatio-temporal data sets has been a topic of recent study, e.g., see~\cite{MSB+22, RHB25}.
Other work has leveraged kPCA to develop tools for the quantitative analysis of the geometry of samples from a probability distribution in a high-dimensional Euclidean space, which is approximately supported on a low-dimensional set, and is corrupted by high-dimensional noise~\cite{LITTLE2017}. 
While we assume the QoI maps defining the design space are specified in this work, a future study will consider the utilization of the OED approaches developed in this work to optimize the data acquisition studies that subsequently lead to data-derived and data-learned QoI.

\section{Scaling Effect}\label{sec:scaling}

Intuitively, we expect that designs exhibiting large sensitivities to variations in the parameters are the most informative and useful in DCI.
For example, suppose we have $m$ measurement all of the same type, which we can configure in any of $K$ ($K>m$) ways to collect data leading to a design space $\Qspace$ of size $K\choose m$.
For a hypothetical measured QoI value, we further assume the uncertainty is quantified in terms of sets whose measures are independent of the specific configuration of measurement devices (i.e., independent of the QoI map). 
This leads to descriptions of uncertainty in a particular output value in terms of sets of a fixed measure.
To determine the optimal design among the $K\choose m$ candidates, we seek the configuration of devices that leads, on average, to the smallest sets of parameters associated with these output sets of fixed size. 
To motivate the use of singular values to quantify the scaling effect, suppose initially that $Q\in\Qspace$ is linear and a bijection so that $m=n$.  
Then, there exists a full-rank $n\times n$ matrix $J$, such that $Q(\lambda)=J\lambda$.
If $\Lambda=\mathbb{R}^n$, and we define $E$ to be an event in the associated data space given by a generalized rectangle, then $Q^{-1}(E)$ is given by a parallelepiped.
From standard results in measure theory and linear algebra, the measure of this parallelepiped is given by
\begin{equation}\label{eq:fudnamental_inverse_scaling}
	\pmeas(Q^{-1}(E)) =  \dmeas(E)\det(J^{-1}) = \dmeas(E)\left(\prod_{k=1}^{m} \, \sigma_{k}\right)^{-1},
\end{equation}
where $\set{\sigma_k}_{1\leq k\leq m}$ are the singular values of $J$.
In the following subsection, we extend this concept to cases where $m<n$ and nonlinear maps.
\subsection{Quantifying Scaling with SVDs}\label{sec:scaling_quantification}
Here, we propose a simple quantitative metric, referred to as the scaling effect, that describes the precision of using a map $Q\in\Qspace$ to identify parameters that map to an event $E$ in the associated data space. 
This metric is easily computed using the singular value decomposition (SVD) of the Jacobian $J_{Q}(\lambda)$ of the map $Q$ at samples of the parameters.

If $m<n$ and $E$ is a generalized rectangle, then $Q^{-1}(E)$ is defined by a cylinder set in $\pspace$ with cross-sections given by $m$-dimensional parallelepipeds. 
The following lemma and ensuing corollary below are used to describe the measures of $m$-dimensional parallelepipeds embedded in $n$-dimensional Euclidean spaces assuming that $m\leq n$.
We are specifically interested in two cases: (i) $m$-dimensional parallelepipeds defined by the rows of a given matrix $J$, and (ii) $m$-dimensional parallelepipeds determined by the cross-sections of $n$-dimensional cylinders given by the pre-image of a $m$-dimensional unit cube under $J$.

\begin{lemma}\label{lem:prod_singvals}
Let $J$ be a full rank $m\times n$ matrix with $m\leq n$, and $Pa(J)$ denote the $m$-dimensional parallelepiped defined by the $m$ rows of $J$. 
The Lebesgue measure $\mu_m$ in $\mathbb{R}^m$ of $Pa(J)$ is given by the product of the $m$ singular values $\set{\sigma_k}_{k=1}^m$ of $J$, i.e., 
\begin{equation}\label{eq:vol_Pa(J)}
    \mu_m(Pa(J)) = \prod_{k=1}^m \sigma_k.
\end{equation}
\end{lemma}

\begin{proof}
    See Appendix~\ref{app:scaling-skewness-proofs}.
\end{proof}

We now turn our attention to the second case of describing the size of a $m$-dimensional parallelepiped defined by the cross-section of an $n$-dimensional cylinder given by the pre-image of a $m$-dimensional unit cube under $J$.
In this case, we utilize the pseudo-inverse of $J$, given by $J^\dagger = J^{\top}(JJ^{\top})^{-1}$.  
As is evident from the formula of the pseudo-inverse, the range of $J^\dagger$ is defined by the row space of $J$.
This observation implies that the $\mu_m$-measure of the cross-section of the pre-image of a unit cube under $J$ is equal to the the $\mu_m$-measure of the parallelepiped defined by the columns of $J^\dagger$.

\begin{corollary}\label{cor:cross_section}
Let $J$ be a full rank $m\times n$ matrix with $m\leq n$.  Then, $Pa((J^\dagger)^{\top})$ is an $m$-dimensional parallelepiped defining a cross-section of the pre-image of an $m$-dimensional unit cube under $J$, and its Lebesgue measure $\mu_m$ is given by the inverse of the product of the $m$ singular values $\set{\sigma_k}_{k=1}^m$ of $J$, i.e., 
\begin{equation}
    \mu_m(Pa((J^\dagger)^{\top})) = \left(\prod_{k=1}^{m} \, \sigma_{k}\right)^{-1}.
\end{equation}
\end{corollary}

\begin{proof}
    See Appendix~\ref{app:scaling-skewness-proofs}.
\end{proof}

Corollary~\ref{cor:cross_section} provides a means to compute the measure of these $m$-dimensional parallelepipeds for a given Jacobian $J_Q(\lambda)$ of a map $Q\in\Qspace$. 
With the goal of describing the precision of using the map $Q\in\Qspace$ to identify parameters that map to $E$ we introduce the following definition.
\begin{definition}\label{def:local-scaling}
For any $Q\in\Qspace$ and letting $J_{Q}(\lambda)\in\mathbb{R}^{m\times n}$ denote the Jacobian of $Q$ at a point $\lambda\in\pspace$ with singular values  $\set{\sigma_k}_{1\leq k\leq m}$, we define the {\bf local scaling effect} of $Q$ as
\begin{equation}\label{eq:scaling_effect}
	SE_Q(\lambda) := \begin{cases}
						\left(\prod_{k=1}^{m} \, \sigma_{k}\right)^{-1}, &  \text{ if } \sigma_k>0, 1\leq k\leq m, \\
						+\infty, & \text{ else.}
					\end{cases}
\end{equation}
\end{definition}

\subsection{Expected Scaling Effect}\label{sec:ESE}


Unlike the inverse sets for linear maps, the measure of an inverse set associated with a nonlinear map depends on the model parameters $\lambda$. 
Specifically, different output events $E$ of the same $\dmeas$-measure may correspond to events in $\pspace$ with different $\pmeas$-measure. 
This parameter dependence of the $\pmeas$-measure of the inverse sets for nonlinear maps motivates an integral-based global scaling effect of a map to define a type of expected scaling effect. 
Before giving this definition, we first address the scenario where the Jacobian of $Q\in\Qspace$ has rank $r<m$ for some $\lambda\in\pspace$.
At such parameters, we formally define $SE_Q(\lambda)=+\infty$. 
In order to both accommodate this scenario and also reduce the impact of statistical outliers (since they are unlikely to be observed in practice) in the definition of the expected scaling effect below, we use the harmonic mean.
Specifically, if $\lambda$ is a random variable and $f(\lambda)>0$, then the harmonic mean of $f$ is given by
\begin{equation*}
	H(f) := \Bigg(\frac{1}{\harmmeas(\pspace)}\int_{\pspace} \frac{1}{f(\lambda)} \, d\harmmeas \Bigg)^{-1},
\end{equation*}
where $\mu^{\text{hm}}_\pspace$ is any measure on $(\pspace,\pborel)$ used to compute the harmonic mean.
We typically take $\mu^{\text{hm}}_\pspace$ to be the volume measure, i.e., $\harmmeas = \pmeas$, or the initial measure, i.e., $\harmmeas = \initmeas$.
With this harmonic mean, we now define the expected scaling effect.
\begin{definition}\label{def:scaling}
For any $Q\in\Qspace$ and $\lambda\in\pspace$, we define the {\bf expected scaling effect (ESE)} as
\begin{equation}\label{Eq:scaling_average}
	ESE(Q) := H(SE_Q(\lambda)) =  \Bigg(\frac{1}{\harmmeas(\pspace)}\int_{\pspace} \frac{1}{SE_Q(\lambda)} \, d\harmmeas\Bigg)^{-1}.
\end{equation}
\end{definition}

Any numerical integration technique can be used to estimate $ESE(Q)$. 
In this paper, we use standard Monte Carlo techniques based on finite sampling to obtain estimates of Jacobians. 
This expected scaling effect will be utilized in Section~\ref{sec:oed} to define an OED formulation.
\section{Skewness Effect}\label{sec:skewness}

Skewness of maps is originally introduced and studied in \cite{MWR_Lindley,Inlet} and refers to the degree to which the interior angles between connected sides in a generalized rectangle $E\in\dborel$ become less perpendicular when transformed into the pre-image set $Q^{-1}(E)$. 
In these other works, it is shown that increasing the skewness of the QoI map, $Q$, directly results in an increased number of samples in $\pspace$ required to accurately estimate the $\pmeas$-measure of $Q^{-1}(E)$ within a given tolerance.
However, in those previous works, no efficient computational approach to estimate the skewness of a map is provided.
In the remainder of this section, we show that, as with scaling, the singular values of the Jacobian map are able to estimate the skewness of a map. 
\subsection{Quantifying Skewness with SVDs}\label{sec:skewness_quantifying}

To more precisely connect skewness to parallelepipeds, suppose initially that $Q$ is linear and $E\in\dborel$ is a generalized rectangle.
We then have that the Jacobian $J=J_Q(\lambda)$ is independent of $\lambda$, and $Q=J_Q\lambda$.
Let $J_{Q,k}$ denote the submatrix of $J_{Q}$ with the $k$th row, denoted by $\mbf{j}_k$, removed and define the map $\widehat{Q} := J_{Q,k}\lambda$, which maps $\pspace$ to a $(m-1)$-dimensional data space $\widehat{\dspace}$. 
Let $P$ denote the projection matrix from $\dspace$ to $\widehat{\dspace}$.
Then, $Q^{-1}(E)$ is a cylinder with cross-sections given by $m$-dimensional parallelepipeds, and $\widehat{Q}^{-1}(P{E})$ is a cylinder with cross-sections given by $(m-1)$-dimensional parallelepipeds.
By construction, $Q^{-1}(E)\subset \widehat{Q}^{-1}(P{E})$, and the $(m-1)$-dimensional parallelepiped cross-sections of the cylinder described by $\widehat{Q}^{-1}(P{E})$ are faces of the $m$-dimensional parallelepipeds defining the cross-sections of $Q^{-1}(E)$.
In other words, the use of all $m$ component maps {\em truncates} the cylinder given by $\widehat{Q}^{-1}(P{E})$ in a particular direction.
To quantify this truncation, we first note that $\mbf{j}_k$ can be decomposed as $\mbf{j}_k=\mbf{j}_k^{0}+\mbf{j}_k^\perp$ where $\mbf{j}_k^{0}$ denotes the projection of $\mbf{j}_k$ into the vector subspace spanned by the other row vectors of $J_Q$, and $\mbf{j}_k^\perp$ denotes the part of $\mbf{j}_k$ that is orthogonal to this projection.
With this notation, the {\em length} of the truncation of the cylinder set $\widehat{Q}^{-1}(PE)$ can be quantified by $\|\mbf{j}_k^\perp\|$ (where $\|\cdot \|$ denotes the Euclidian norm), which conceptually corresponds to the amount of additional, or unique, information contained in the $k$th observed component used to define $Q$ compared to all the other components of $Q$. 

In the more general case where $Q$ is nonlinear, we use the Jacobian of $Q$ to define the local skewness of the map $Q\in\Qspace$ as originally introduced in \cite{Inlet}. 
\begin{definition}\label{def:skew}
For any $Q\in\Qspace$, $\lambda\in\pspace$, and $1\leq k\leq m$, denote by $\mbf{j}_k$ the $k$th row of the Jacobian $J_{Q}(\lambda)$. 
The skewness vector, $\mbf{SK_Q}(\lambda)$, is an $m$-dimensional vector-valued function whose $k$th component is given by
	\begin{equation}\label{Eq:skew}
        [\mbf{SK_Q}(\lambda)]_k := \begin{cases} \frac{\|\mbf{j}_k\|}{\|\mbf{j}_k^{\perp}\|}, & \|\mbf{j}_k^{\perp}\|\neq 0, \\
        								\infty, & \text{else}, 
        								\end{cases} \ 1\leq k\leq m.
	\end{equation}
    Then, we define the {\bf local skewness} of the map $Q\in\Qspace$ at a point $\lambda$ as
	\begin{equation}\label{Eq:Skew}
        SK_Q(\lambda) = \norm{\mbf{SK_Q}(\lambda)}_\infty,
	\end{equation}
 
\end{definition}

In the above definition, the $k$th component of the skewness vector, $[\mbf{SK_Q}(\lambda)]_k$, describes the amount of {\em redundant} information present in the $k$th component of the observable map compared to what is present in the other $m-1$ components when constructing the pre-image near the point $\lambda\in\pspace$. 
This observation implies that the smallest value any component of this vector can take is one, which is associated with a Jacobian of $Q$ containing orthogonal rows.

While the above definition is arrived at via a geometrical argument, it is not entirely practical from a computational perspective given the reliance on decomposing many vectors in distinct subspaces to determine their orthogonal components. 
However, the geometric argument is based on parallelepipeds embedded in the parameter space.
With this perspective, the following fundamental decomposition result proves critical to computing skewness with SVDs.

\begin{theorem}\label{thm:fundamental_decomp}
Let $J$ be a full rank $m\times n$ matrix with $m\leq n$, and $Pa(J)$ denote the $m$-dimensional parallelepiped defined by the $m$ row vectors $\set{\mbf{j}_1, \hdots, \mbf{j}_m}\subset\mathbb{R}^n$ of $J$. 
There exists $\mbf{j}_1^{\perp}, \mbf{j}_1^0\in\mathbb{R}^n$ such that
    \[
        \mbf{j}_1 = \mbf{j}_1^{\perp} + \mbf{j}_1^0, \hskip 10pt \mbf{j}_1^{\perp}\perp \mbf{j}_1^0, \hskip 10pt \mbf{j}_1^0\in span\{\mbf{j}_2, \hdots, \mbf{j}_m\},    
    \]
    and
	\[
        \mu_m(Pa(J)) = \|\mbf{j}_1^{\perp}\| \times \mu_{m-1}(Pa(J_1)),
	\]
	where 
	\begin{itemize}
		\item $J_1$ denotes the submatrix of $J$ with the first row removed,  
		\item $Pa(J_1)$ is the $(m-1)$-dimensional parallelepiped defined by the $m-1$ row vectors of $J_1$, and
		\item $\mu_m$ and $\mu_{m-1}$ represent the $m$- and $(m-1)$-dimensional Lebesgue measures, respectively. 
	\end{itemize}
\end{theorem}

\begin{proof}
    The proof is provided in \cite{Crampin_Pirani}.
\end{proof}

The following result summarizes the method for computing the skewness in terms of singular values of the Jacobian of $Q$.

\begin{corollary}\label{Cor:singular_values}
    For any $Q\in\Qspace$, let $\mbf{j}_k$ denote the $k$th row of the Jacobian $J_{Q}(\lambda)$, then
    	\begin{equation}
    		SK_Q(\lambda) =  \max_{1\leq k\leq m} \, 	\frac{\|\mbf{j}_k\|\prod_{i=1}^{m-1}\sigma_{k,i}}{\prod_{k=1}^{m}\sigma_k},
    \end{equation}
 where $\set{\sigma_k}_{1\leq k\leq m}$ and $\set{\sigma_{k,i}}_{1\leq i\leq m-1}$ denote, respectively, the singular values of the Jacobian $J_Q({\lambda})\in\mathbb{R}^{m\times n}$ and the submatrix $J_{Q,k}(\lambda)\in\mathbb{R}^{(m-1)\times n}$ defined by omitting the $k$th row from $J_{Q}(\lambda)$.
\end{corollary}

\begin{proof}
    See Appendix~\ref{app:scaling-skewness-proofs}.
\end{proof}


\subsection{Expecting Skewness Effect}

Since $SK_Q(\lambda)$ may vary substantially over $\pspace$, we must quantify this variability in order to optimally choose $Q\in\Qspace$.
There is no largest value since there exists maps $Q$ such that the condition of the Jacobian may be arbitrarily large.
If the Jacobian were to ever fail to be full rank, then $SK_Q(\lambda)$ is infinite since there exists rows of the Jacobian where $\|\mbf{j}_k^{\perp}\|=0$, which would occur if some of the components of the QoI map can, locally, be replaced by linear combinations of other components. 
Thus, as with the scaling effect, we use the harmonic mean in the following definition. 
\begin{definition}\label{def:avgskew}
For any $Q\in\Qspace$, we define the {\bf average}  (or {\bf expected}) skewness of $Q$ as
\begin{equation}\label{Eq:skewint}
	ESK(Q) := \Bigg(\frac{1}{\harmmeas(\pspace)}\int_{\pspace} \frac{1}{SK_Q(\lambda)} \, d\harmmeas\Bigg)^{-1},
\end{equation}
\end{definition}
where we again typically take $\harmmeas=\pmeas$ or $\harmmeas = \initmeas$.


\section{Defining an OED with ESE or ESK}\label{sec:oed}

OED uses a model to estimate the utility of an experiment both before the experiment is conducted and assuming that the model can be used to predict the data with some uncertainty.
Typically, a utility function is maximized to determine the OED, such as the expected information gain, as mentioned in the introduction.
We take this perspective in the formal definition given below of the OED for DCI as a maximization problem.
\begin{definition}[OED for DCI]\label{def:oed}
  Let $\Qspace$ the space of all possible experimental designs, $Q\in\Qspace$ be a specific design, and $U$ denote a utility function. Then, the OED is the $Q\in\Qspace$ that maximizes the utility of the experiments (i.e., potential QoI maps). Formally, we write this OED as
  \begin{equation}\label{eq:qopt}
    Q_{\text{opt}} := \arg \max_{Q\in\Qspace}U(Q).
  \end{equation}
\end{definition}
In this work, the ESE or ESK of a map should be minimized to determine optimality, so we define utility functions for these properties based on their reciprocals, i.e.,
\begin{equation}\label{eq:utility}
U(Q) = ESE^{-1}(Q), \quad \text{or} \quad U(Q) = ESK^{-1}(Q).
\end{equation}
In Section~\ref{subsec:simulOED}, we demonstrate that using the reciprocal of the ESE or the ESK leads to intuitive optimal designs, and in Section~\ref{subsec:OED_DCI} we show that these design choices have a significant impact on the solutions to nominal stochastic inverse problems.

The focus in this paper is on the utilization of certain geometric properties, namely scaling or skewness of Jacobians, for OED for the purpose of DCI.
For the sake of focusing on these utility functions, we do not explore different approaches for solving the optimization problem given by Definition~\ref{def:oed} and simply find the optimal design over a discrete set of candidate designs.
In other words, we discretize $\Qspace$ to define a finite design space in all problems and leave the topic of continuous optimization over $\Qspace$ to a future study.

\section{Skewness as incremental scaling and a greedy algorithm}\label{subsec:greedy}


In this section, we provide a theoretical interpretation of skewness as quantifying incremental changes in scaling.
Specifically, we show that computing skewness is equivalent to computing the relative change in scaling of the proposed $m$-dimensional map compared to all of its $(m-1)$-dimensional counterparts.
This makes skewness a natural choice for the design criterion of a greedy algorithm.


\begin{corollary}\label{Cor:Skewness_as_scaling_update}
	For any $Q\in\Qspace$, 
	\begin{equation}
		SK_Q(\lambda) = \frac{SE_Q(\lambda)}{\| \operatorname{diag}(\|\mbf{j}_k\|) \mbf{SE_{Q,m-1}} \|_\infty},
	\end{equation}
	where $\operatorname{diag}(\|\mbf{j}_k\|)\in\mathbb{R}^{m\times m}$ denotes the diagonal matrix whose $k$th diagonal component is the Euclidean norm of the $k$th row of the Jacobian $J_{Q}(\lambda)$, and $\mbf{SE_{Q,m-1}}\in\mathbb{R}^m$ denotes the vector whose $k$th component is the scaling effect of the $(m-1)$-dimensional QoI map defined by removing the $k$th component from $m$-dimensional map $Q$.
\end{corollary}

\begin{proof}
    This follows immediately from Definition~\ref{def:scaling} and Corollary~\ref{Cor:singular_values}.
\end{proof}

In many cases, it is infeasible to exhaustively search the design space $\Qspace$ according to a design criterion to determine the OED.  
For example, suppose we have nine sensors that can be placed in any of 100 possible locations.
In this case, the design space is defined by ${{100}\choose{9}} = 1.9$E12 possible maps. 
Assuming the design criteria defining the utility function in the OED problem takes only $1$E-6 seconds to evaluate, the process of determining the OED will take an estimated $527.7$ days to compute.  
Therefore, a greedy algorithm for constructing vector-valued maps component-by-component has proven to be an attractive alternative in such scenarios, e.g., see~\cite{Huan_Jagalur_Marzouk_2024,Drovandi02012014,Dror01032008}.
In this work, we propose the use of skewness as the design criteria in a greedy algorithm for two reasons.
First, at a conceptual level, skewness quantifies the amount of new information a proposed additional measurement has compared to existing measurements making it particularly well-suited for a greedy algorithm.
Second, it has a mathematical connection to how scaling is transformed when moving from an $(m-1)$-dimensional QoI map to an $m$-dimensional map as summarized by Corollary~\ref{Cor:Skewness_as_scaling_update}.




We first compute the $ESE^{-1}(Q)$ for a design space, denoted by $\Qspace^{(1)}$, defined by all potential {\em scalar}-valued maps. 
The optimal design based on the scaling criteria is used to define the first component of the QoI map.
We then compute the $ESK^{-1}(Q)$ for a new design space, denoted by $\Qspace^{(2)}$, of potential two-dimensional QoI maps whose first component is fixed from the first step.
Again, the (greedy) OED is determined by the $Q\in\Qspace^{(2)}$ with the greatest $ESK^{-1}(Q)$.  
The algorithm continues in this way until either (1) the dimension of the vector-valued map matches the dimension of the desired data space $\mathcal{D}$ specified at the beginning of the algorithm, or (2) $ESK^{-1}(Q)$ is less than some prescribed tolerance for all $Q$ in the current design space being optimized.
The latter case suggests that there are no further QoI components to add among possible candidates that bring sufficiently new (i.e., non-redundant) information compared to the existing lower-dimensional QoI map.
This is summarized in Algorithm~\ref{alg:greedy} where we let $\Qspace^{(d)}$ denote a design space defined by vector-valued observable maps of dimension $d$ where the first $d-1$ components are chosen by the previous steps in the algorithm.

\begin{algorithm}\caption{A Greedy Algorithm for OED}\label{alg:greedy}
\begin{algorithmic}
    \STATE {Specify desired $m=\dim(\dspace)$ and $tol>0$.}
    \STATE {Let $\Qspace^{(1)}$ denote all possible scalar maps and $Q_\text{opt} = \arg \max_{Q\in\Qspace^{(1)}} ESE^{-1}(Q)$.}
    \STATE {Let $\Qspace^{(2)} = \set{ (Q_\text{opt}, Q) \, : \, Q\in\Qspace^{(1)} }$.}
\FOR{$d = 2, \dots, m$}
    \STATE {Let $Q_\text{opt} = \arg \max_{Q\in\Qspace^{(d)}} ESK^{-1}(Q)$}
    \IF {$d<m$ and $\exists$ $Q\in\Qspace^{(d)}$ such that $ESK^{-1}(Q)\geq tol$}
	\STATE {Let $\Qspace^{(d+1)} = \set{ (Q_\text{opt}, Q) \, : \, Q\in\Qspace^{(1)} }$.}
    \ELSE
        \STATE {Exit for-loop.}
    \ENDIF
\ENDFOR
\end{algorithmic}
\end{algorithm}

\section{Numerical Results}
\label{S:Example}

In this section we define a simulation model and use it to test and compare the various design criteria considered in this work.
The model represents a common scenario in computational science and engineering where we seek to characterize an intrinsic variability over a population of manufactured components.
This problem formulation is well-suited for the DCI framework, and in our first numerical example, we compare the solutions to representative stochastic inverse problems for (locally) optimal designs and suboptimal designs. 



\subsection{Model setup}\label{S:model-setup}

Consider the problem of heating a population of thin metal rod with uncertain thermal conductivity.
For the mathematical model, we use the time-dependent diffusion equation to model the temperature $u$ as a function of space-time coordinate $(x,t)$:
\begin{equation}\label{eq:cbayes_heatrod}
	\begin{cases}
    	\rho c\frac{\partial u}{\partial t} = \nabla \cdot (\kappa \nabla u) + S, & x \in \Omega,  t \in (t_0, t_f], \\
    	\frac{\partial u}{\partial n} = 0 , & x \in \partial \Omega,  t \in (t_0, t_f], \\
    	 u(x) = 0, & x \in \Omega, t=t_0,
    \end{cases}
\end{equation}
where $\Omega = (0, 1)$ represents a metal alloy rod of unit length, $\rho=1.5$ is the density of the rod, $c=1.5$ is the heat capacity, $\kappa$ is the (uncertain) thermal conductivity, $t_0=0$, $t_f=1$, and $S$ is a Gaussian source defined by
\begin{equation*}
	S(x) = 50\exp\Big(-\frac{(0.5-x)^2}{0.05}\Big). 
\end{equation*}

Suppose the rod is manufactured by welding together two rods of equal length and of similar alloy type.
For the sake of simplicity, we do not attempt to characterize the material properties of the weld itself.
However, due to the manufacturing process, the actual alloy compositions may differ across the population, which implies that the thermal conductivities satisfy
\begin{equation*}
    \kappa = \begin{cases}
        \lambda_1, & x < 0.5 \\
        \lambda_2, & x \geq 0.5
    \end{cases},
\end{equation*}
where $\lambda_1$ and $\lambda_2$ are uncertain and may vary between 0.01 and 0.2. 
We emphasize that for an individual rod, the uncertainty in the material properties and the associated specific values of thermal conductivities is considered epistemic, which could be characterized and quantified using other inference technique.
However, when the population is considered, the uncertainty in the material properties and associated thermal conductivities is aleatoric since there is intrinsic and irreducible variability across the population.
Thus, we set $\pspace=[0.01, 0.2]^2$.
Given a sample of thermal conductivity values from $\pspace$, we approximated solutions to the PDE in Eq.~\eqref{eq:cbayes_heatrod} using the open source finite element software 
MrHyDE~\cite{mrhyde_user,mrhyde2023github}
with 40 uniform piecewise linear finite elements and a second-order implicit midpoint time-integration rule with $t_0=0<t\leq 1.0=t_f$ and 20 uniform time steps.
The experiments simulated placing two contact thermometers that record two separate temperature measurements at time $t=t_f$.  
The points $(x_0, x_1)\in\Omega\times\Omega$ index the $Q\in\Qspace$.  

For computational simplicity, we used the finite element spatial discretization of the physical domain $\Omega$ into $41$ equally spaced points to discretize $\Qspace$. 
Specifically, we let 
\[\{(x_0^{(k)}, x_1^{(k)})\}^{820}_{k=1}\subset \Omega\times\Omega,\]
denote a point from $\Omega\times\Omega$ restricted to the unique subset of the possible tensor products of points forming the finite element mesh, 
which defines an indexing of potential observable maps as
\[\{Q_{k}\}_{k=1}^{820}, \text{ where } Q_{k}:=\big(u(x_0^{(k)}, t_f), u(x_1^{(k)}, t_f)\big).\]  
We note that there is symmetry within $\Omega\times\Omega$ along the line $x_0 = x_1$,
so we only need to consider plots and optimization of the unique criteria since the others correspond to a re-ordering of the contact thermometers.

\subsection{Solution space characteristics}
In the left plot of Fig.~\ref{fig:test_problem_solutions}, we show some representative solutions $u(x,t_f)$ for various parameter values taken from $\pspace$ to help build intuition for solutions of the OED problems under various design criteria.
When thermal conductivity is low ($\kappa=0.01$) on both sides of the rod, the temperature increases more near the center of the rod where the source is located.
If the left and right sides of the rod have thermal conductivities of 0.01 and 0.2, respectively, then the temperature profile is continuous, but non-differentiable, at the center of the rod.
Moreover, since heat is conducted more easily on the right side of the rod, the temperature near the insulated right endpoint of the rod increases faster than on the left side of the rod where the temperature will increase significantly more near the source than near the endpoint.
Finally, when the thermal conductivity is high ($\kappa=0.2$) on both sides of the rod, the temperature profile is qualitatively similar to when the thermal conductivity is low on both sides of the rod, but, a quantitative comparison reveals that the temperature at the center is reduced by about 60\%  while the temperatures at the insulated endpoints are about 600\% higher compared to the case where thermal conductivities are low on both sides.

\begin{figure}[htbp]
\centering
{\includegraphics[width=0.45\textwidth]{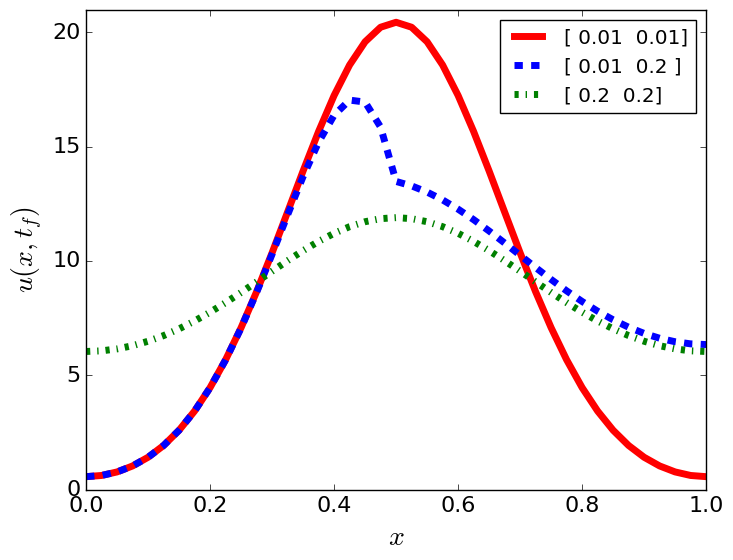}}
{\includegraphics[width=0.45\textwidth]{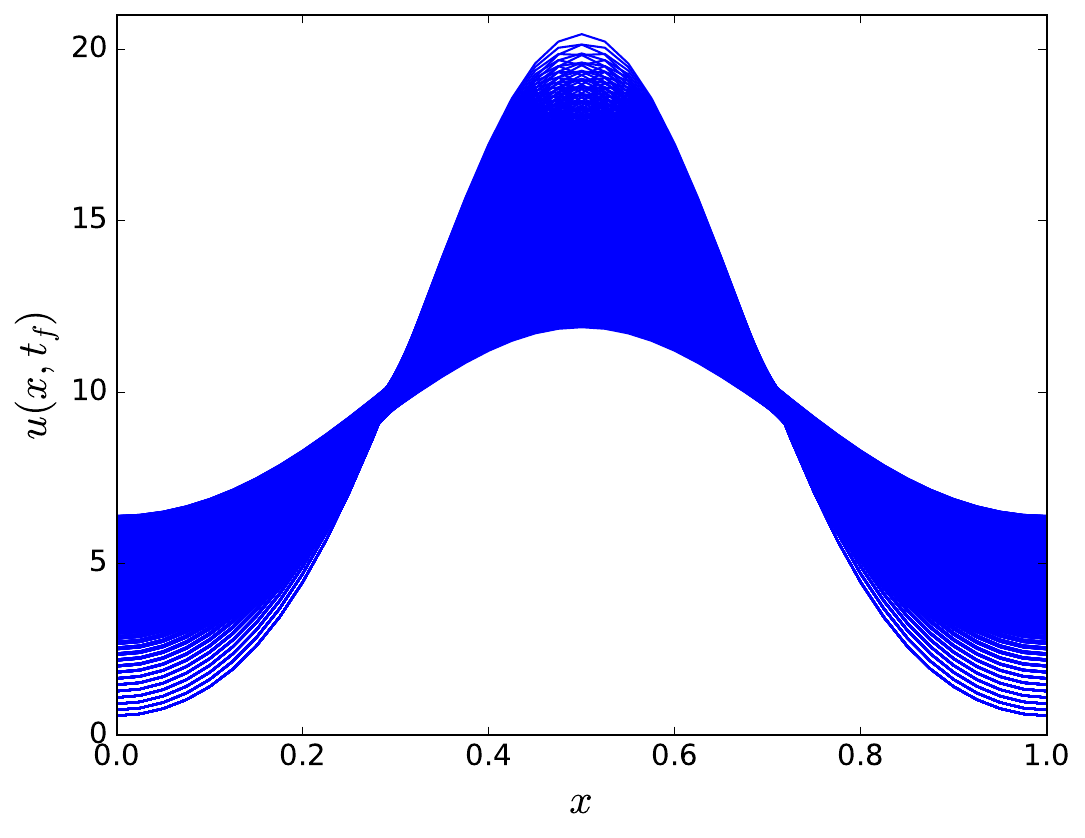}}
\caption{Left: Temperature profiles $u(x,t_f)$ for thermal conductivity of the left- and right-sides of the rod are respectively 0.01 and 0.01 (solid red curve), 0.01 and 0.2 (dashed blue curve), and 0.2 and 0.2 (dashed-dotted green curve).
Right: Temperature profiles $u(x,t_f)$ for a uniformly regular set of $50\times 50$ samples of thermal conductivities.}
\label{fig:test_problem_solutions}
\end{figure}

The right plot of Fig.~\ref{fig:test_problem_solutions} illustrates the variability in the temperature profiles over a set of $50\times 50$ uniformly spaced samples of thermal conductivity in $\pspace$.
There are two clear demarcation points in the plot where most of the temperatures agree regardless of the thermal conductivity.
Specifically, we observe that around $x=0.3$ and $x=0.7$, the variation in temperatures as a function of thermal conductivity is relatively small compared to other spatial locations.
This small variation influences all of the OED results shown below since measuring the solution at these points does not define QoI data that is useful for identifying the material properties.


\subsection{Singular Values and OED}\label{subsec:simulOED}

In this section, we optimize the design criteria defined by $ESE^{-1}(Q)$ or $ESK^{-1}(Q)$ for the model introduced in Section~\ref{S:model-setup}.
Specifically, we generated 10000 samples 
\[
\{\lambda^{(i)}\}_{i=1}^{10000}\subset\pspace,
\]
and computed the Jacobian matrices 
\[
\big\{\{J_{Q_k}(\lambda^{(i)})\}_{i=1}^{10000}\big\}_{k=1}^{820}
\] 
using a forward finite difference method with a perturbation of 1.0e-5.
Since the number of outputs greatly exceeds the number of inputs, an adjoint-based approach to compute the derivatives is not practical.  
The values of 
\[\{ESE^{-1}({Q_k})\}_{k=1}^{820} \quad \text{and} \quad \{ESK^{-1}({Q_k})\}_{k=1}^{820}\] 
were approximated by Monte Carlo estimates using the 10000 random samples of the thermal conductivities.

Since $\Qspace$ is indexed by the points in $\Omega\times\Omega$, we show a plot of $ESE^{-1}$ over $\Omega\times\Omega$ in Fig.~\ref{fig:mt_heatrod_scaling}.
We note that there are multiple local maxima, which are denoted by black dots.
In the table of Fig.~\ref{fig:mt_heatrod_scaling}, we provide the utility function values for each of these local maxima.
The local maxima at $(1.0,0.0)$ is associated with placing a thermometer at each end of the rod, which the authors believe corresponds to an intuitive placement of thermometers.
However, we observe that the other local maxima are associated with QoI maps that have comparable utility values (within $2\%$ of the one indexed by $(1.0,0.0)$).

In the plot of Fig.~\ref{fig:mt_heatrod_skewness}, we show $ESK^{-1}$ as a function over $\Omega \times\Omega$.
As with the $ESE^{-1}$ utility function, we observe that this also possesses several local maxima.
One of these local maxima again corresponds to placing a thermometer at each end of the rod.
In the table of Fig.~\ref{fig:mt_heatrod_skewness}, we provide the values of $ESK^{-1}$ for each of the local maxima.
In this case, the OED is associated with $(1.0,0.0)$.

\begin{figure}[htbp]
\centering
\begin{minipage}[t]{.5\textwidth}
\centering
\vspace{0pt}
\includegraphics[width=\textwidth]{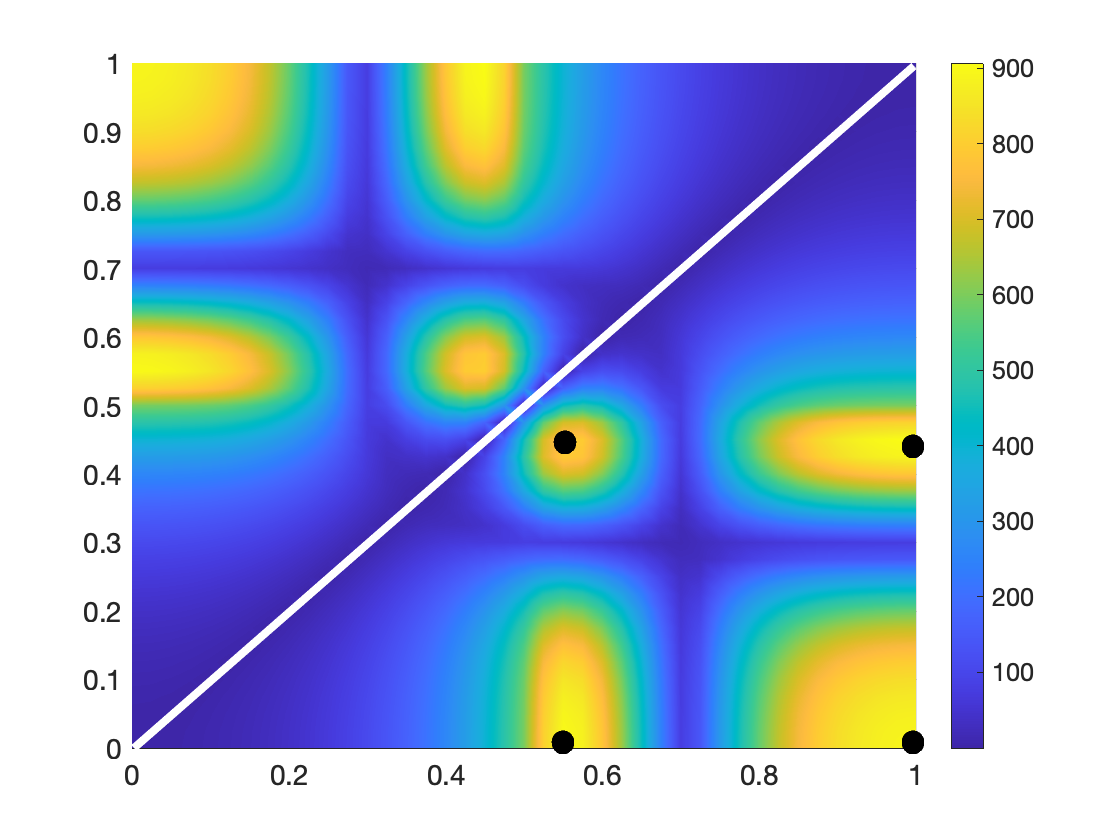}
\end{minipage}
\hskip 5pt
\begin{minipage}[t]{.4\textwidth}
\centering
\vspace{40pt}
\begin{tabular}{ccc}
\hline
Design Point & $ESE^{-1}$ \\ \hline \hline
$$(1.0, 0.45)$$ & $906.9$ \\
$$(0.55, 0.0)$$ & $900.9$ \\
$(1.0, 0.0)$ & $890.5$ \\
$(0.55, 0.45)$ & $801.5$ \\ \hline
\end{tabular}
\end{minipage}
\caption{(Left) Plot of $ESE^{-1}$ over the indexing set $\Omega\times\Omega$ for the design space $\Qspace$.  Note the symmetry across the diagonal.  Focusing on the lower triangular region, the four local maxima are marked with block circles.  (Right) The $ESE^{-1}$ values for the QoI maps associated with each of the four local maxima.}\label{fig:mt_heatrod_scaling}
\end{figure}

\begin{figure}[htbp]
\centering
\begin{minipage}[t]{.5\textwidth}
\centering
\vspace{0pt}
\includegraphics[width=\textwidth]{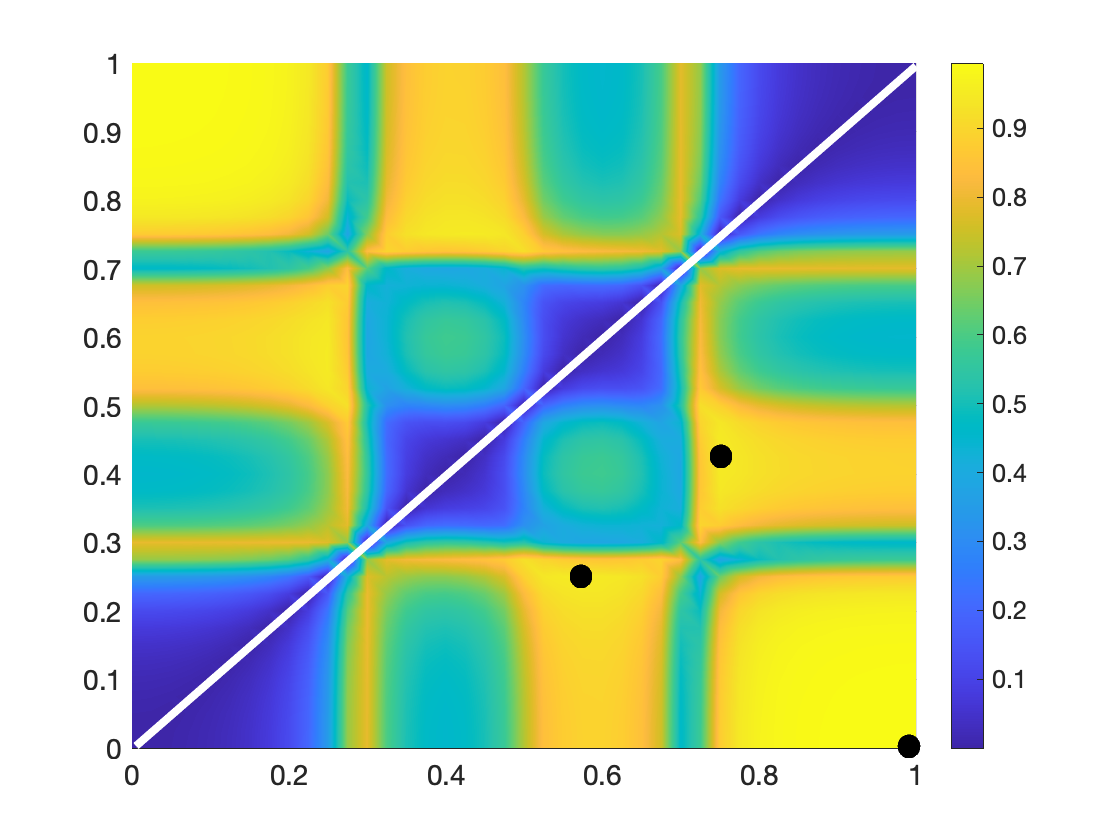}
\end{minipage}
\hskip 5pt
\begin{minipage}[t]{.4\textwidth}
\centering
\vspace{40pt}
\begin{tabular}{ccc}
\hline
Design Point & ${ESK}^{-1}$ \\ \hline \hline
$(1.0, 0.0)$ & 0.995 \\
$(0.575, 0.25)$ & 0.950 \\
$(0.75, 0.425)$ & 0.950 \\ \hline
\end{tabular}
\end{minipage}
\caption{(Left) Plot of $ESK^{-1}$ over the indexing set $\Omega\times\Omega$ for the design space $\Qspace$.  Note the symmetry across the diagonal.  Focusing on the lower triangular region, the most prominent three local maxima are marked with black circles.  (Right) The $ESK^{-1}$ for the QoI associated with each of these three prominent local maxima.}\label{fig:mt_heatrod_skewness}
\end{figure}


\subsection{Impact on Solutions to Stochastic Inverse Problems}\label{subsec:OED_DCI}

Recall that the evaluation of the OED criteria based on $ESE^{-1}$ and $ESK^{-1}$ does not require solving inverse problems.  
However, this work is also motivated by the construction of pullback probability measures within the DCI framework to determine which experiments are likely to be useful for solving the stochastic inverse problem.
To that end, we construct DCI solutions to illustrate the impact of the QoI chosen from this criteria.

For demonstration purposes, we evaluated the model at the midpoint of the input parameter space, and constructed an observed density by centering a Gaussian distribution around this point in the QoI space, using a fixed variance of
$0.15 \mathbb{I}_{2}$ where $\mathbb{I}_{2}$ denotes the $2\times 2$ identity matrix.
As shown in Figs.~\ref{fig:dci_oed1}-\ref{fig:dci_suboed2} below, the geometry of the output spaces associated with the optimal designs are significantly different both quantitatively and qualitatively from the geometry of the output spaces associated with a suboptimal design.
To give a representative sampling of the characteristics of the geometry of the prediction spaces and the associated updated densities, we consider two designs that correspond to local maxima of $ESE^{-1}$ or $ESK^{-1}$, and two designs that correspond to sub-optimal values for these criteria.

Fig.~\ref{fig:dci_oed1} illustrates the  optimal design at (1.0, 0.0) by the black dot in the plots of the $ESE^{-1}$ (upper left) and the $ESK^{-1}$ (upper right).
The predicted QoI samples (lower left) and the corresponding updated density (lower right) associated with this particular combination of sensor locations and observed density are also shown in this figure.
\begin{figure}[htbp]
\centering
\includegraphics[width=0.45\textwidth]{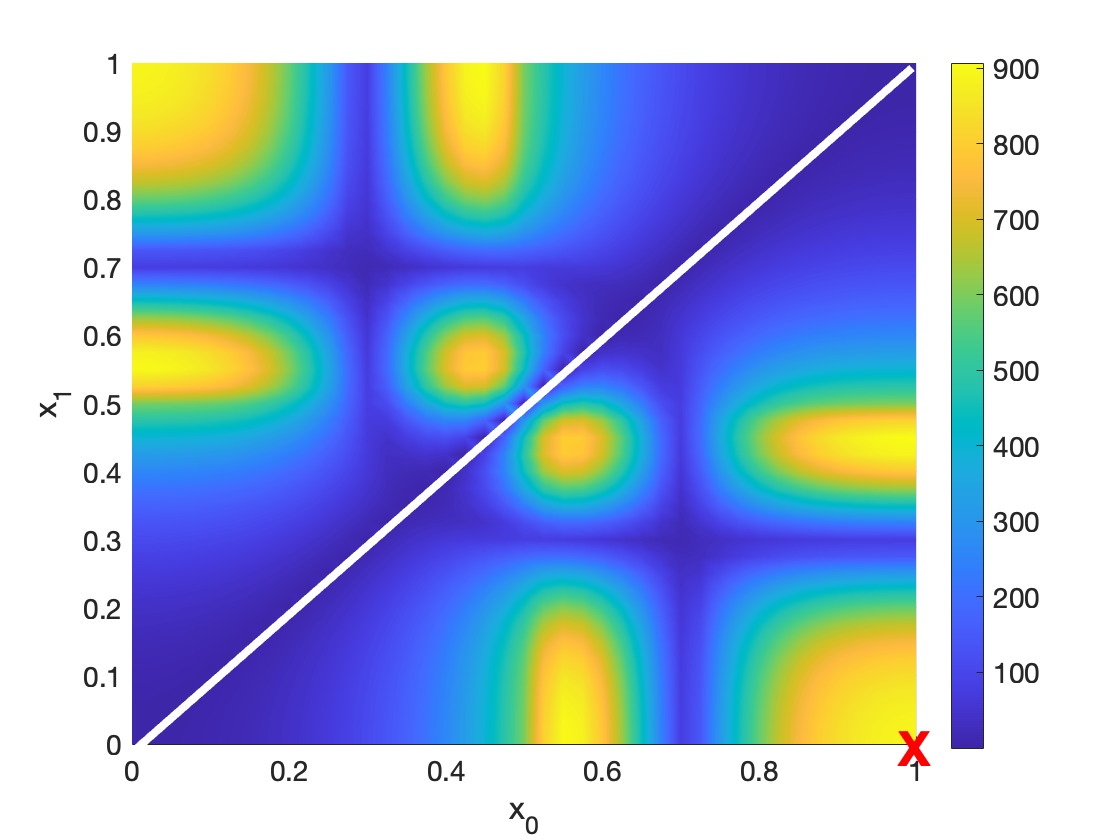}
\includegraphics[width=0.45\textwidth]{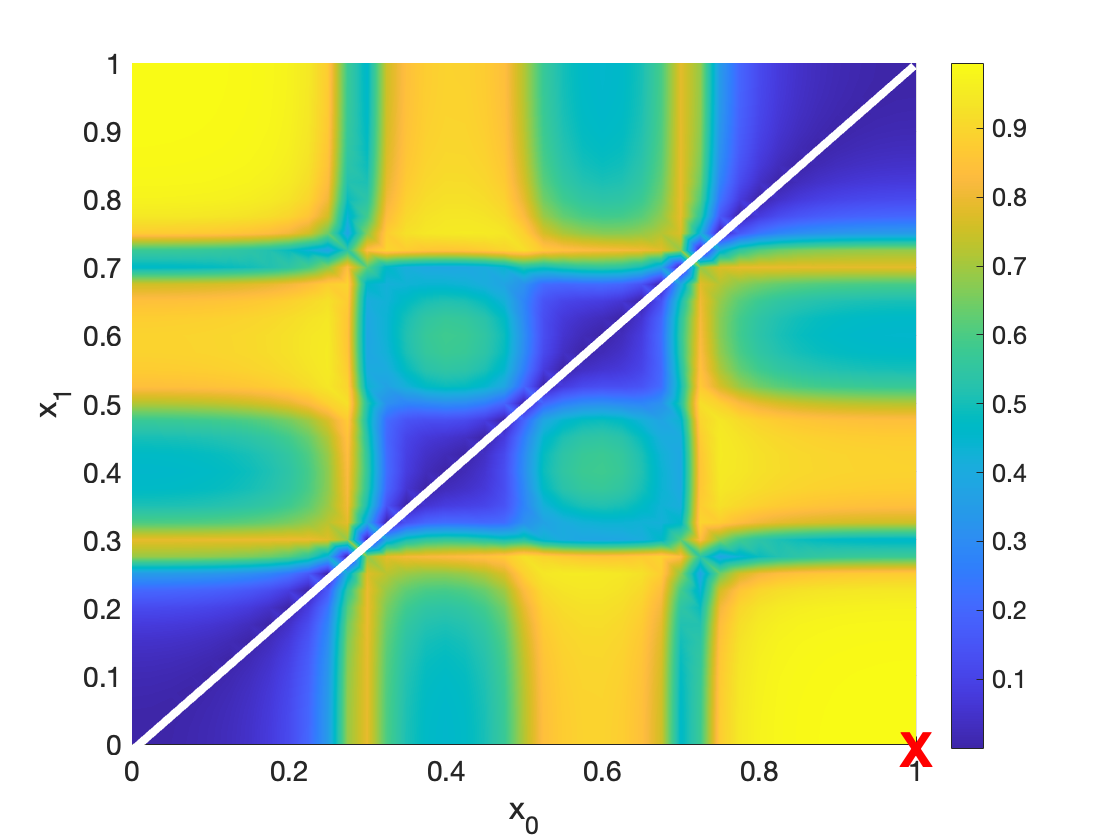}
\includegraphics[width=0.45\textwidth]{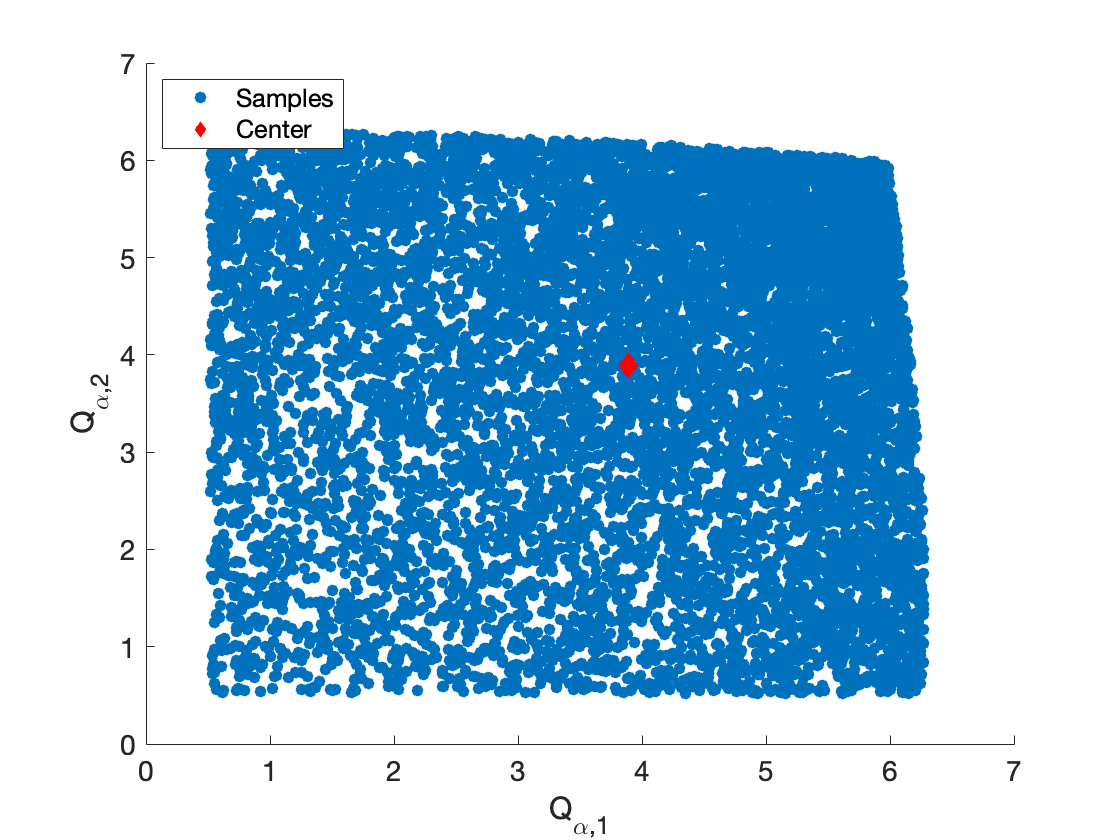}
\includegraphics[width=0.45\textwidth]{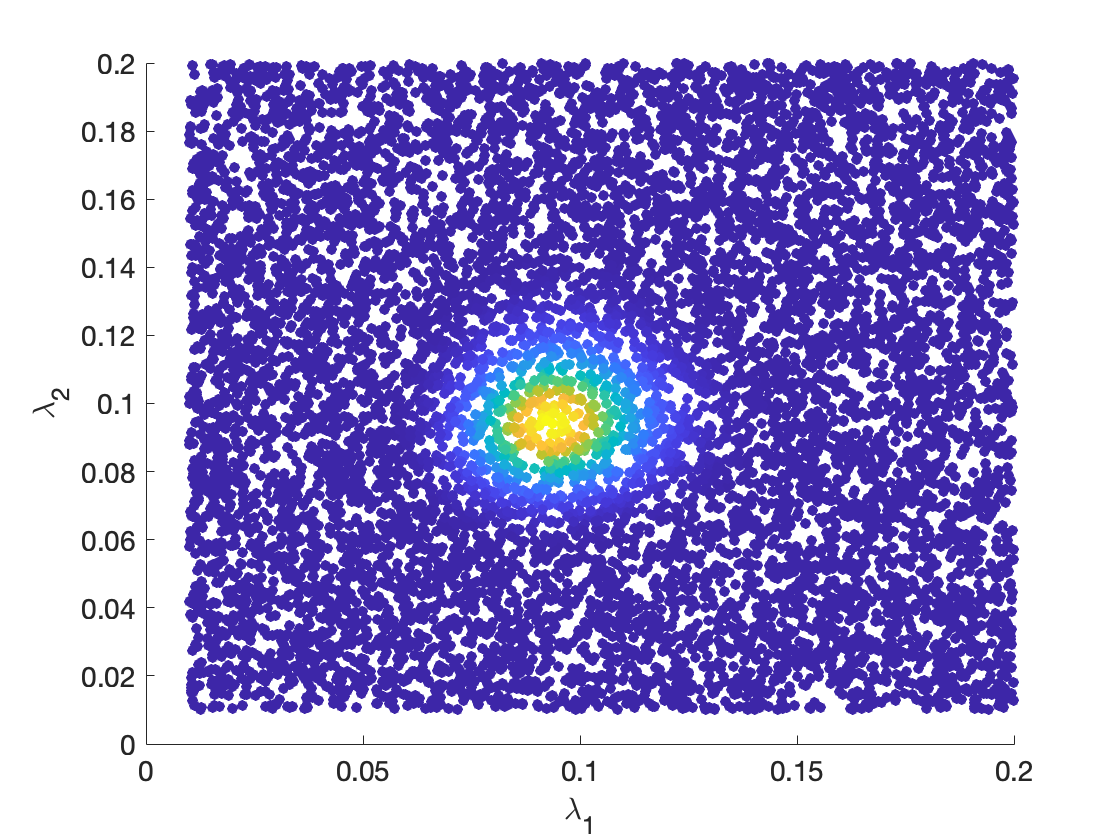}
\caption{The design corresponding to sensors at $x_0=1.0$ and $x_1=0.0$ is illustrated by the red ``X'' in the plots of $ESE^{-1}$ (upper left) and $ESK^{-1}$ (upper right). 
The predicted QoI samples (lower left) and the corresponding updated density (lower right) show the impact of this design. The red diamond in the lower left plot corresponds to the QoI evaluated at the midpoint of the parameter space and is used to define the mean of the Gaussian observed density with a covariance of $0.15 \mathbb{I}_{2}$. }
\label{fig:dci_oed1}
\end{figure}
The impact of this locally optimal design is seen in both the predicted QoI data as well as in the updated density.
Specifically, the prediction space for QoI data is approximately rectangular, which indicates that the two chosen measurements are providing information that are nearly geometrically orthogonal from each other across all of the parameter space.
Subsequently, the updated density is concentrated in a relatively small area around the midpoint that was propagated to construct the observed density.

Next, we consider the impact of the design at (1.0, 0.45), which is a locally optimal design for $ESE^{-1}$, as seen in the upper left plot in Fig.~\ref{fig:dci_oed2}.
However, this design has an inverse $ESK$ value that is approximately 20\% from the previous design considered at $(1.0, 0)$ (upper right plot in Fig.~\ref{fig:dci_oed2}).
The impact of this increase in skewness is immediately seen in the predicted QoI samples (lower left) and the corresponding updated density (lower right).
\begin{figure}[htbp]
\centering
\includegraphics[width=0.45\textwidth]{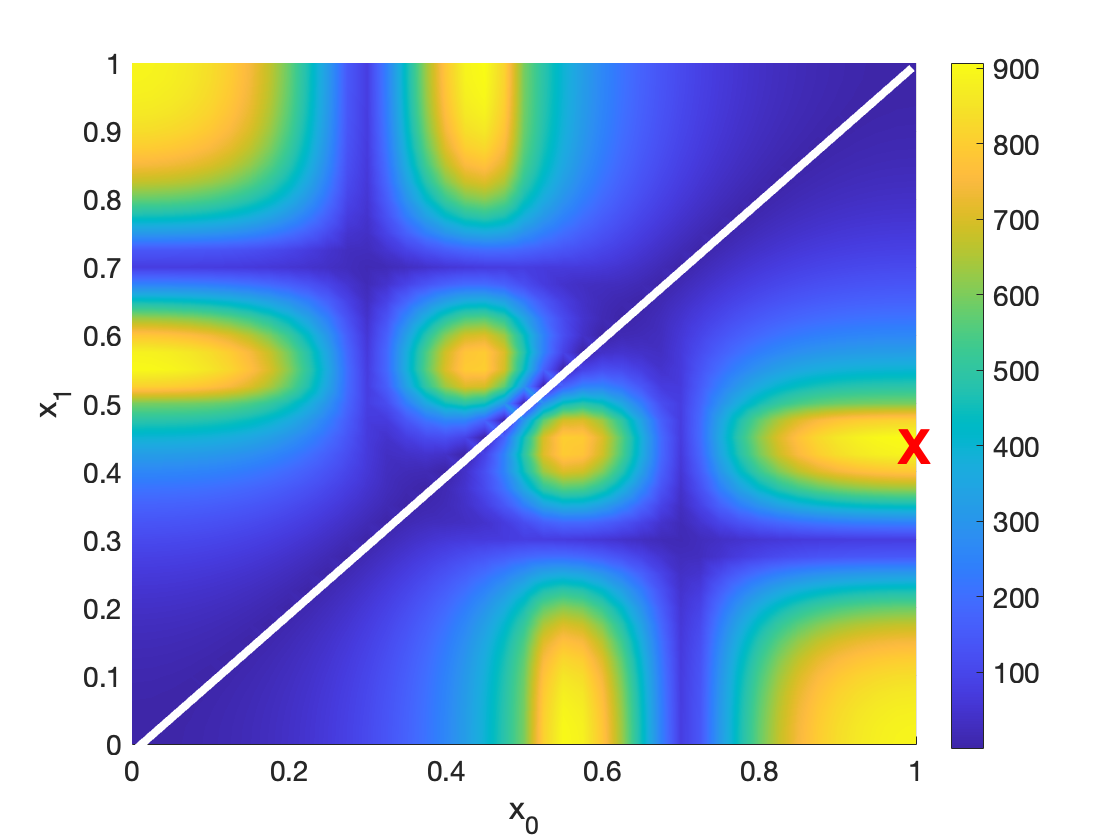}
\includegraphics[width=0.45\textwidth]{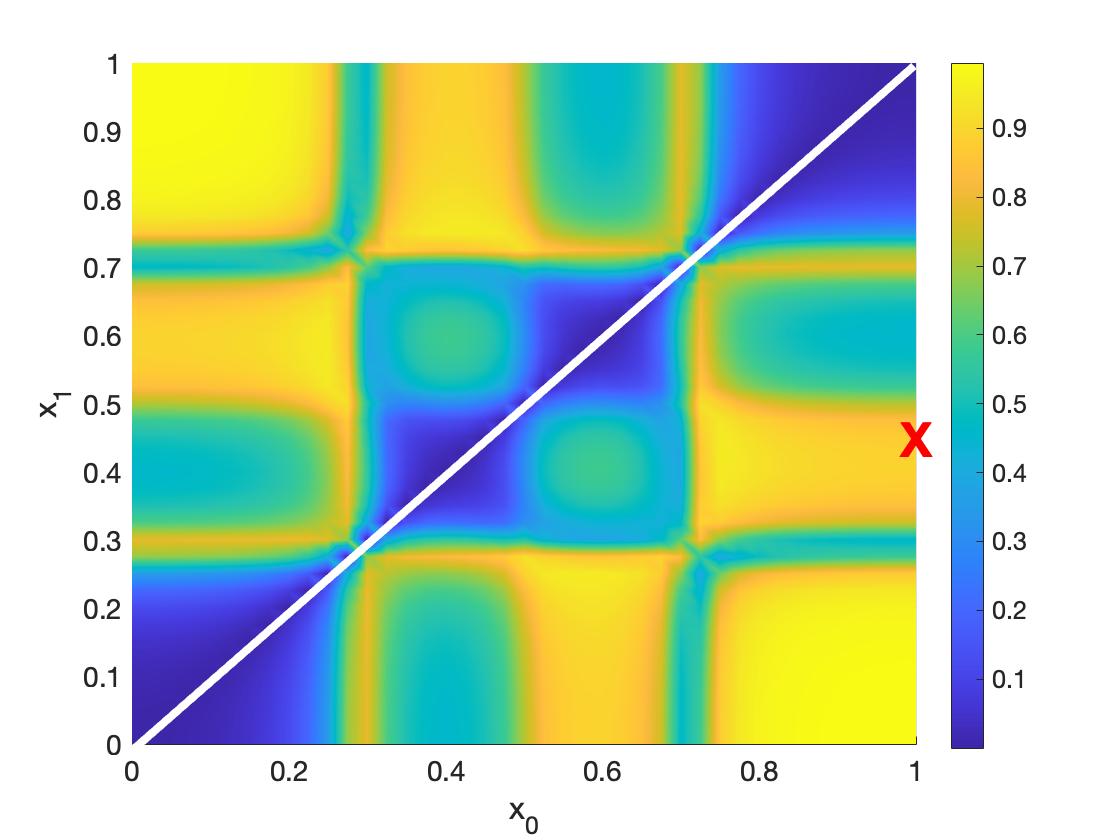}
\includegraphics[width=0.45\textwidth]{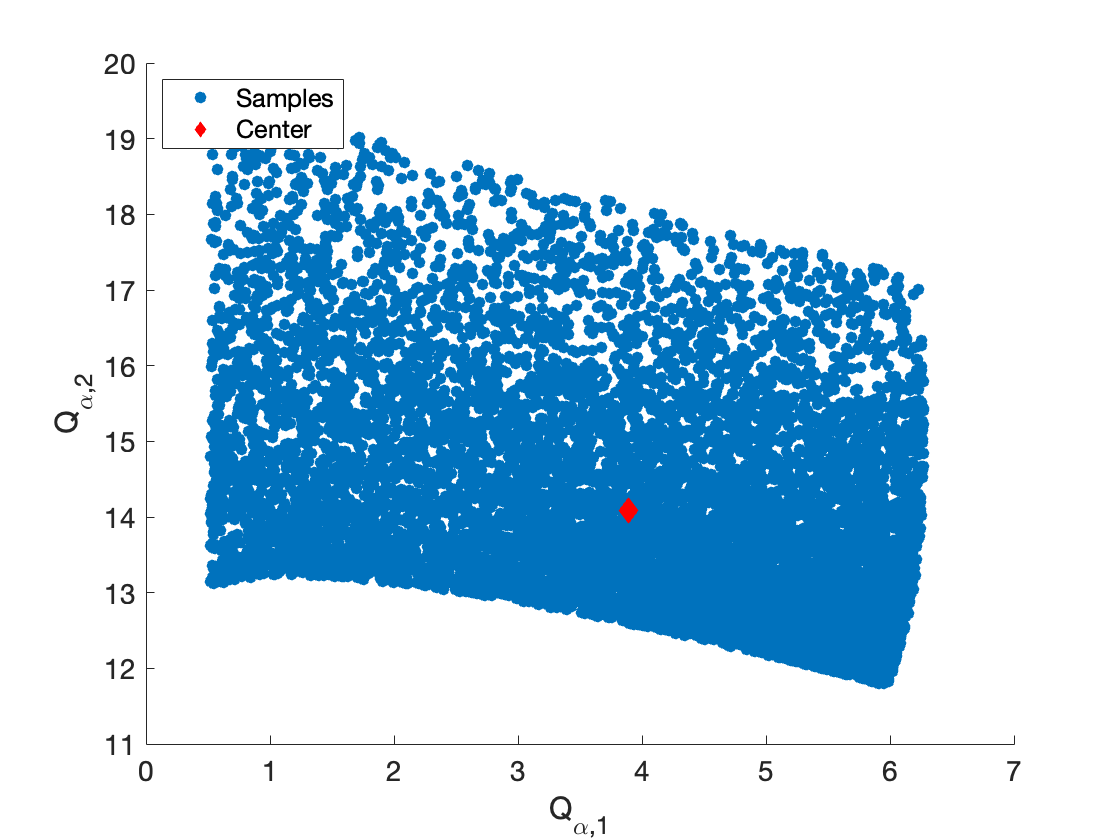}
\includegraphics[width=0.45\textwidth]{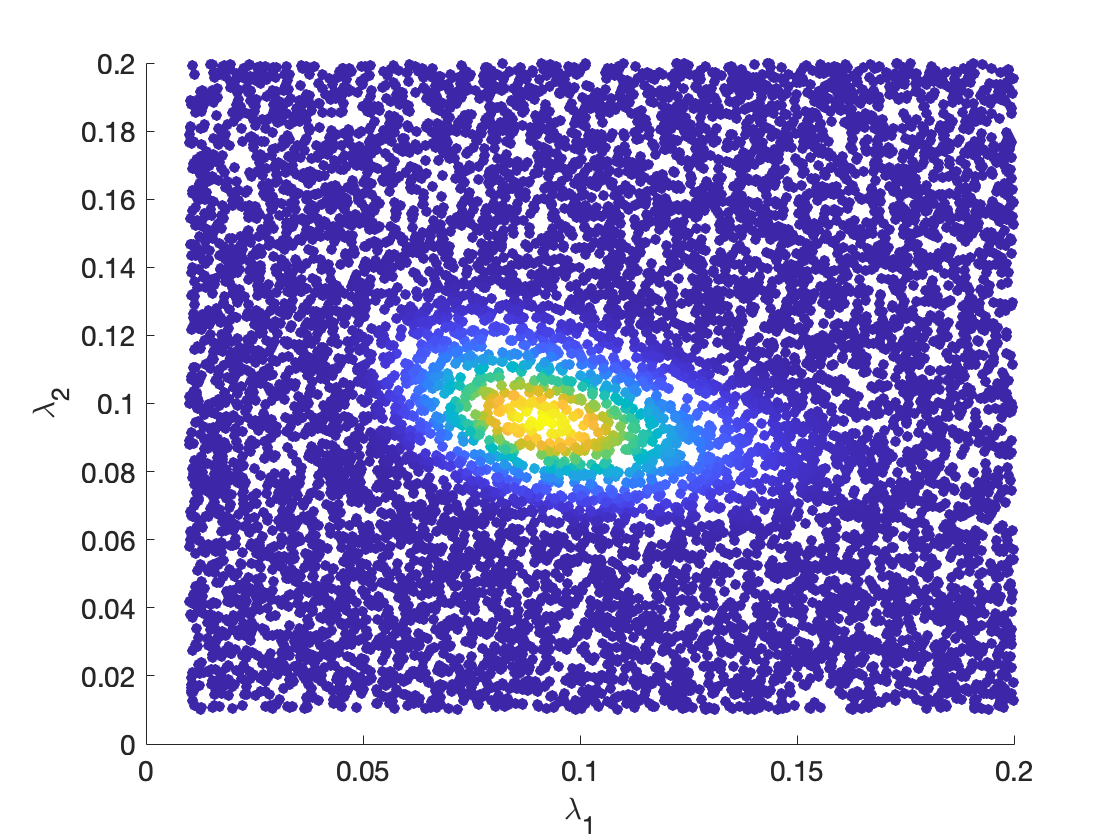}
\caption{The design corresponding to sensors at $x_0=1.0$ and $x_1=0.45$ is illustrated by the red ``X'' in the plots of the $ESE^{-1}$ (upper left) and $ESK^{-1}$ (upper right).  
The predicted QoI samples (lower left) and the corresponding updated density (lower right) show the impact of this design. The red dot in the lower left plot corresponds to the QoI evaluated at the midpoint of the parameter space and is used to define the mean of the Gaussian observed density with a covariance of $0.15 \mathbb{I}_{2}$.}
\label{fig:dci_oed2}
\end{figure}
More precisely, while the prediction space is still approximately rectangular, there is a clear rotation to this space along with a more significant clustering of samples visible in the lower right region compared to the clustering of samples seen in the upper right of the prediction space for the prior design.  
This subsequently produces an updated density that is still relatively concentrated around the midpoint used to construct the observed density but with a more visibly elliptical and skewed covariance structure in comparison to the updated density previously seen in Fig.~\ref{fig:dci_oed1}.

The impact of the next design we consider is at $(0.175, 0.0)$ shown as the black dot in the plots of $ESE^{-1}$ (upper left) and $ESK^{-1}$ (upper right) of Fig.~\ref{fig:dci_suboed1}.
For this design, both the $ESE^{-1}$ and $ESK^{-1}$ are relatively small (i.e., this is a clearly sub-optimal design).
The impact is immediately visible on both the predicted QoI samples (lower left) and the corresponding updated density (lower right).
\begin{figure}[htbp]
\centering
\includegraphics[width=0.45\textwidth]{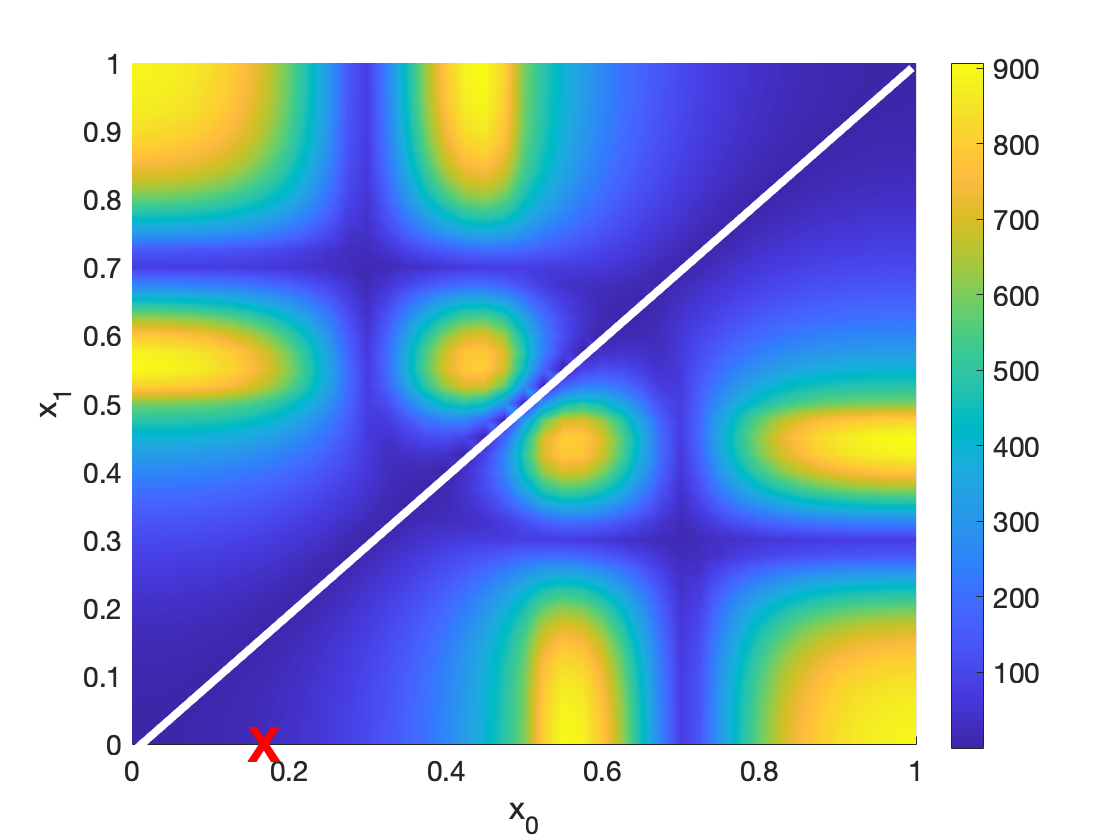}
\includegraphics[width=0.45\textwidth]{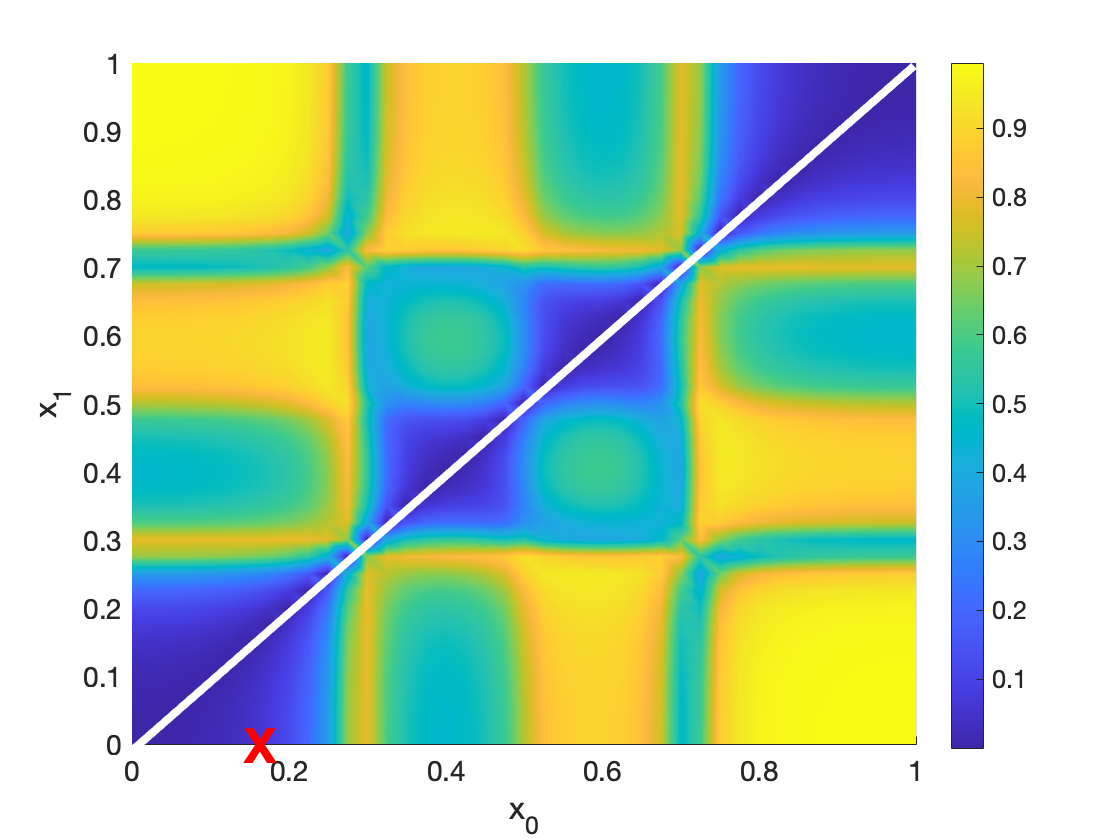}
\includegraphics[width=0.45\textwidth]{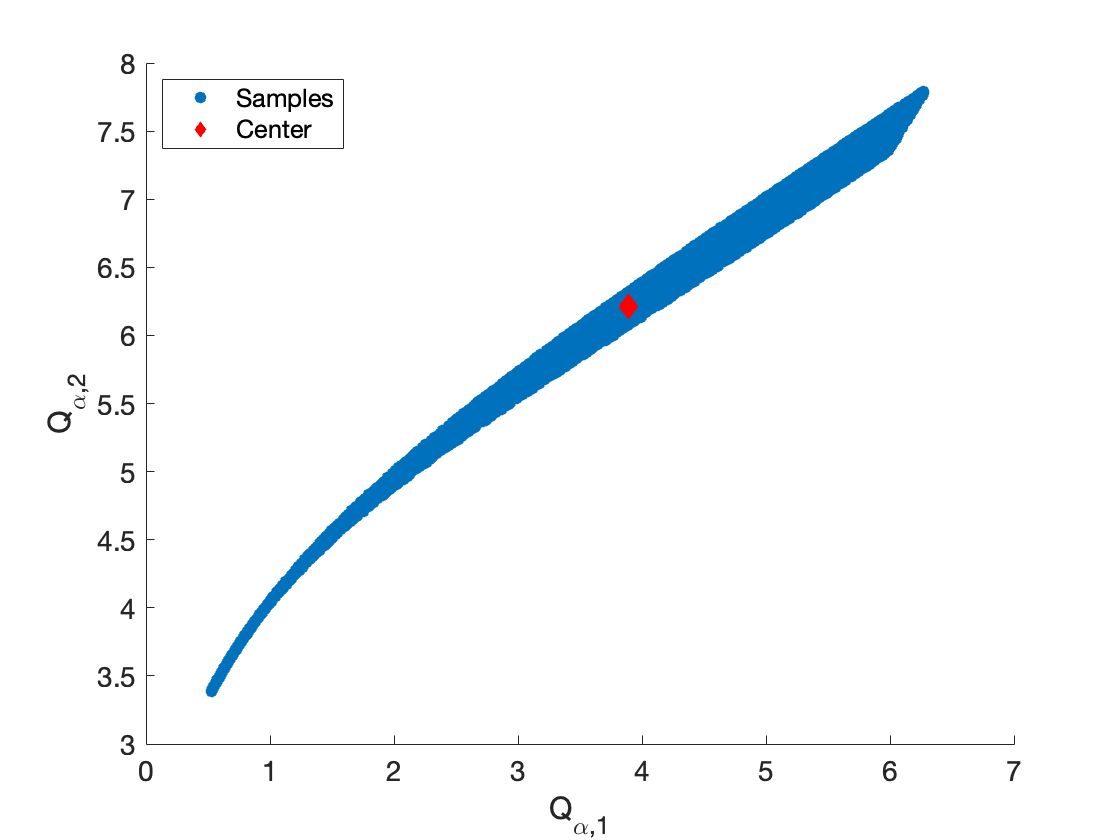}
\includegraphics[width=0.45\textwidth]{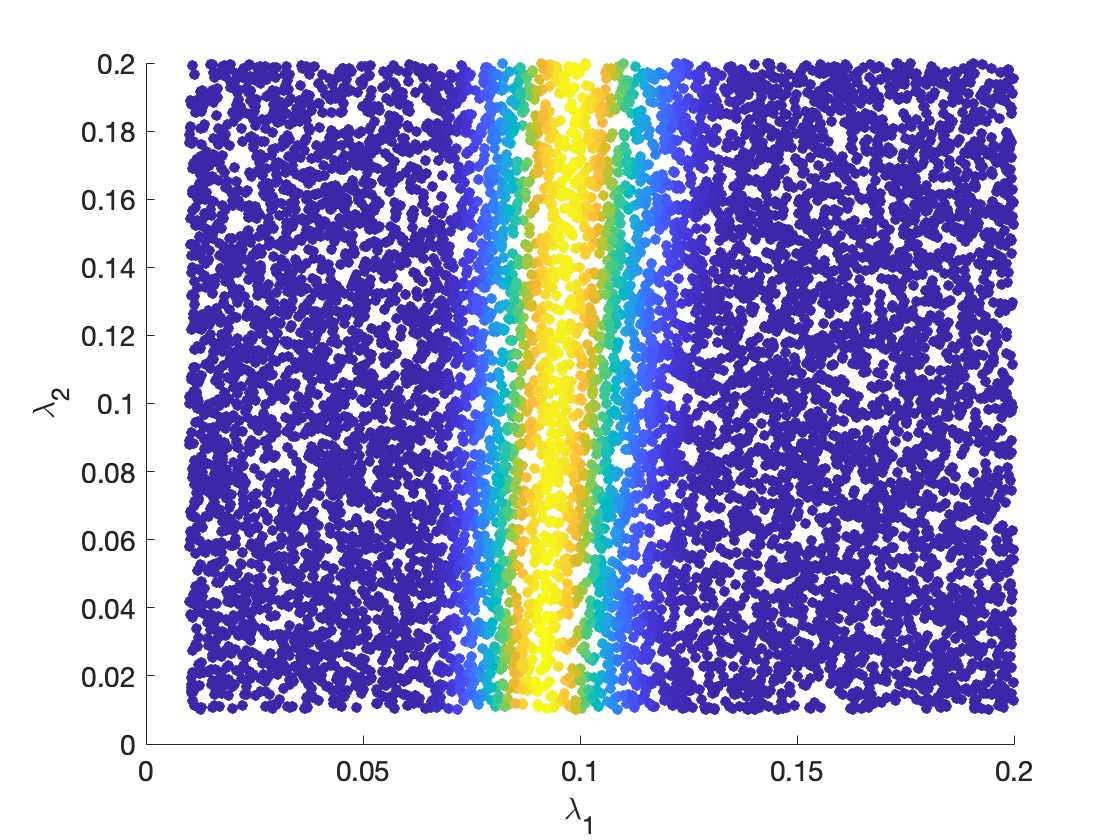}
\caption{The design corresponding to sensors at $x_0=0.175$ and $x_1=0.0$ is illustrated by the red ``X'' in the plots of the $ESE^{-1}$ (upper left) and $ESK^{-1}$ (upper right).  
The predicted QoI samples (lower left) and the corresponding updated density (lower right) show the impact of this design. The red dot in the lower left plot corresponds to the QoI evaluated at the midpoint of the parameter space and is used to define the mean of the Gaussian observed density with a covariance of $0.15 \mathbb{I}_{2}$.}
\label{fig:dci_suboed1}
\end{figure}
For this design, the sensor measurements define highly correlated QoI values resulting in an elongated and narrow prediction space as seen in the lower left of Fig.~\ref{fig:dci_suboed1}.  
The updated density is seen to have a large covariance in the $\lambda_2$-direction (the vertical direction) due to the lack of significant geometric information present in the components of this QoI map for this particular parameter. 

Finally, we consider the design at $(0.7, 0.3)$ which is once again shown as a black dot in the plots of the $ESE^{-1}$ (upper left) and the $ESK^{-1}$ (upper right) of Fig.~\ref{fig:dci_suboed2}.
For this design, the $ESE^{-1}$ value is relatively small while the $ESK^{-1}$ value is around 0.5.
\begin{figure}[htbp]
\centering
\includegraphics[width=0.45\textwidth]{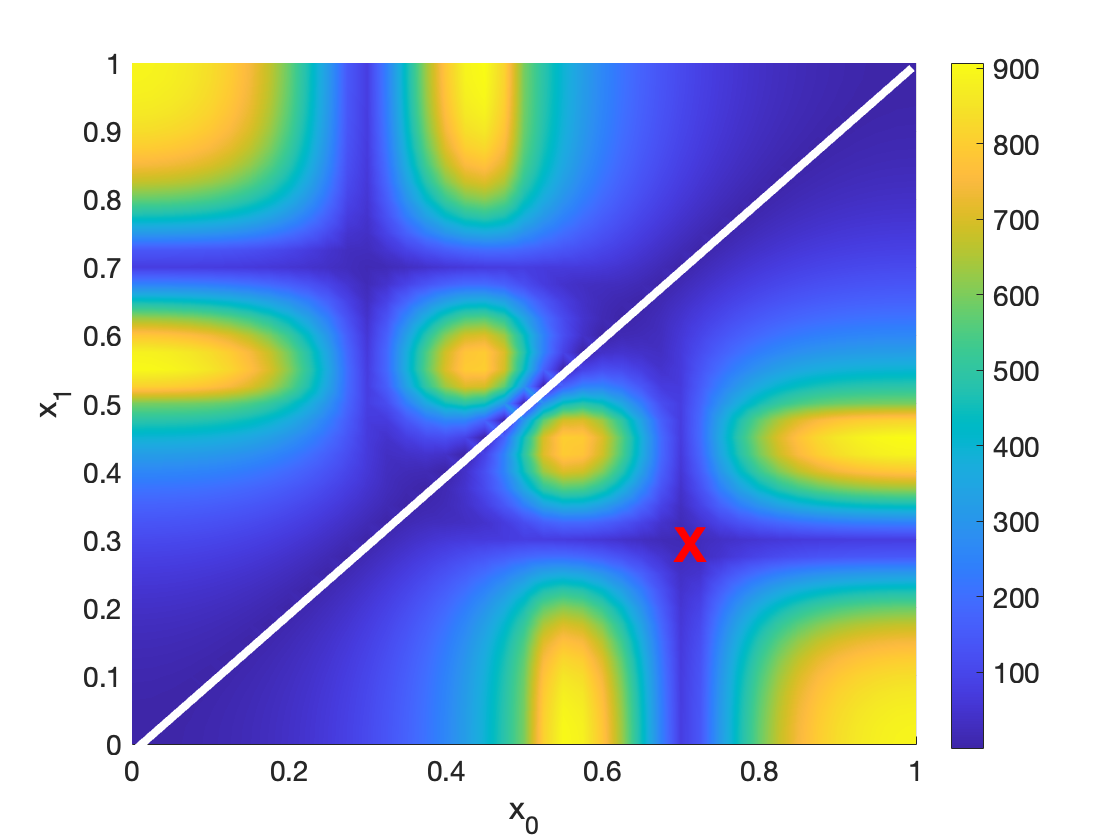}
\includegraphics[width=0.45\textwidth]{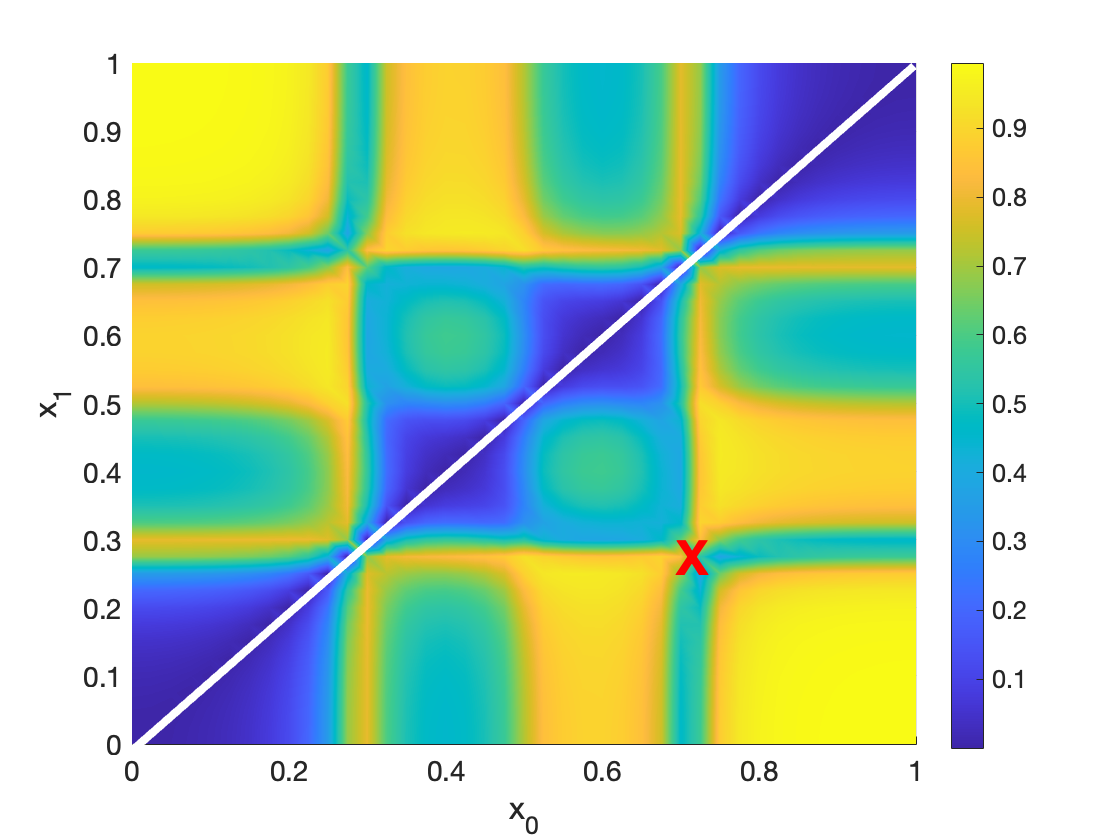}
\includegraphics[width=0.45\textwidth]{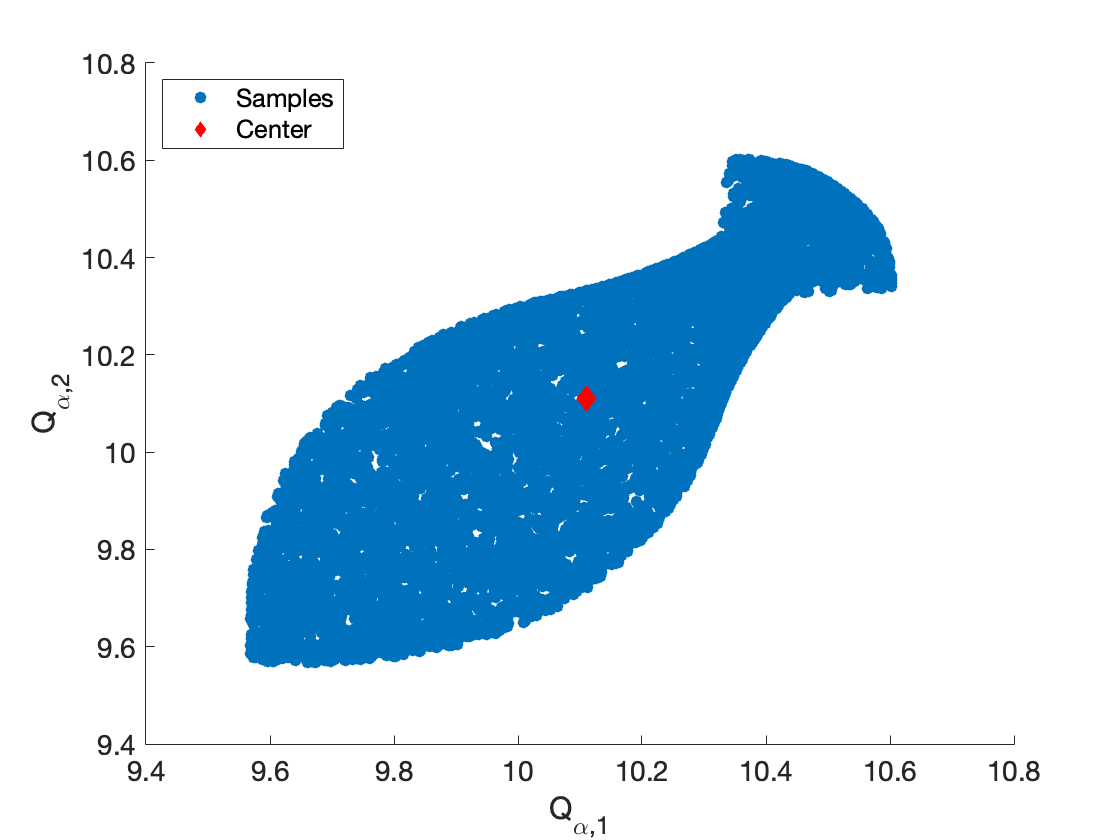}
\includegraphics[width=0.45\textwidth]{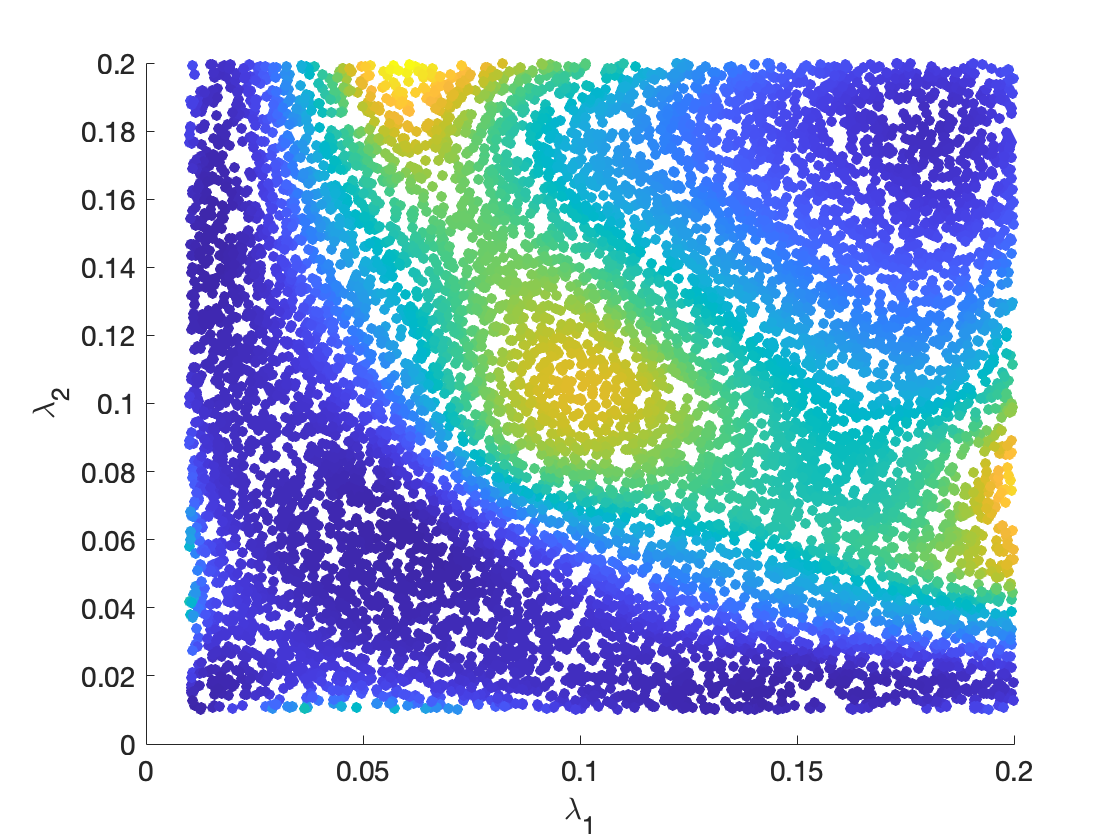}
\caption{The design corresponding to sensors at $x_0=0.7$ and $x=0.3$ is illustrated by the red ``X'' in the plots of the $ESE^{-1}$ (upper left) and $ESK^{-1}$ (upper right).  
The predicted QoI samples (lower left) and the corresponding updated density (lower right) show the impact of this design. The red dot in the lower left plot corresponds to the QoI evaluated at the midpoint of the parameter space and is used to define the mean of the Gaussian observed density with a covariance of $0.15 \mathbb{I}_{2}$.}
\label{fig:dci_suboed2}
\end{figure}
The predicted QoI samples (lower left) and the corresponding updated density (lower right) show the impact of this sub-optimal design.
The prediction space clearly has a more complex structure due to the nonlinear relationship between the QoI collected at these sensors with varying regions of correlation between the components of the QoI map.  
Correspondingly, the updated density in Fig.~\ref{fig:dci_suboed2} has a complex, multimodal structure.

The results in this section clearly demonstrate that choosing optimal or nearly optimal QoI maps for DCI based on either the $ESE^{-1}$ or $ESK^{-1}$ utility functions can have significant impacts on both the predicted QoI space as well as the associated updated densities.
These results are limited to the case where it is possible to optimize the entire QoI map (i.e., all of the component maps) simultaneously.
In the next section, we consider the greedy approach to incrementally optimize each new component of the QoI map. 

\section{Greedy Results Using ESK}
\label{sec:greedy_results}

We now utilize a model with a more complex problem with a higher-dimensional parameter space to illustrate the usefulness of the greedy OED algorithm defined in Algorithm~\ref{alg:greedy}.
We modify the problem in~\eqref{eq:cbayes_heatrod} by setting $\Omega = [0.0, 1.0]^2$ to model a two-dimensional metal plate, and we set $t_f=2$.
We also define the source $S(x)$ as the tensor product of the previous source in both dimensions but re-scaled to have a magnitude of 50. 
Here, we assume that the metal plate is manufactured by welding together nine square plates of the same shape and of similar alloy type.   
As before, we assume that due to the manufacturing process, the actual alloy compositions may differ slightly, leading to uncertain thermal conductivity properties ($\kappa$) in each of the nine regions of the plate, denoted by $\lambda = [\lambda_1, \dots, \lambda_9]$.
These regions and the associated parameters are shown in Fig.~\ref{fig:heatplate_omega}. 
We assume the same ranges of parameters as before leading to $\pspace=[0.01, 0.2]^9$.
\begin{figure}[htbp]
    \centering
    \scalebox{0.45}{\includegraphics{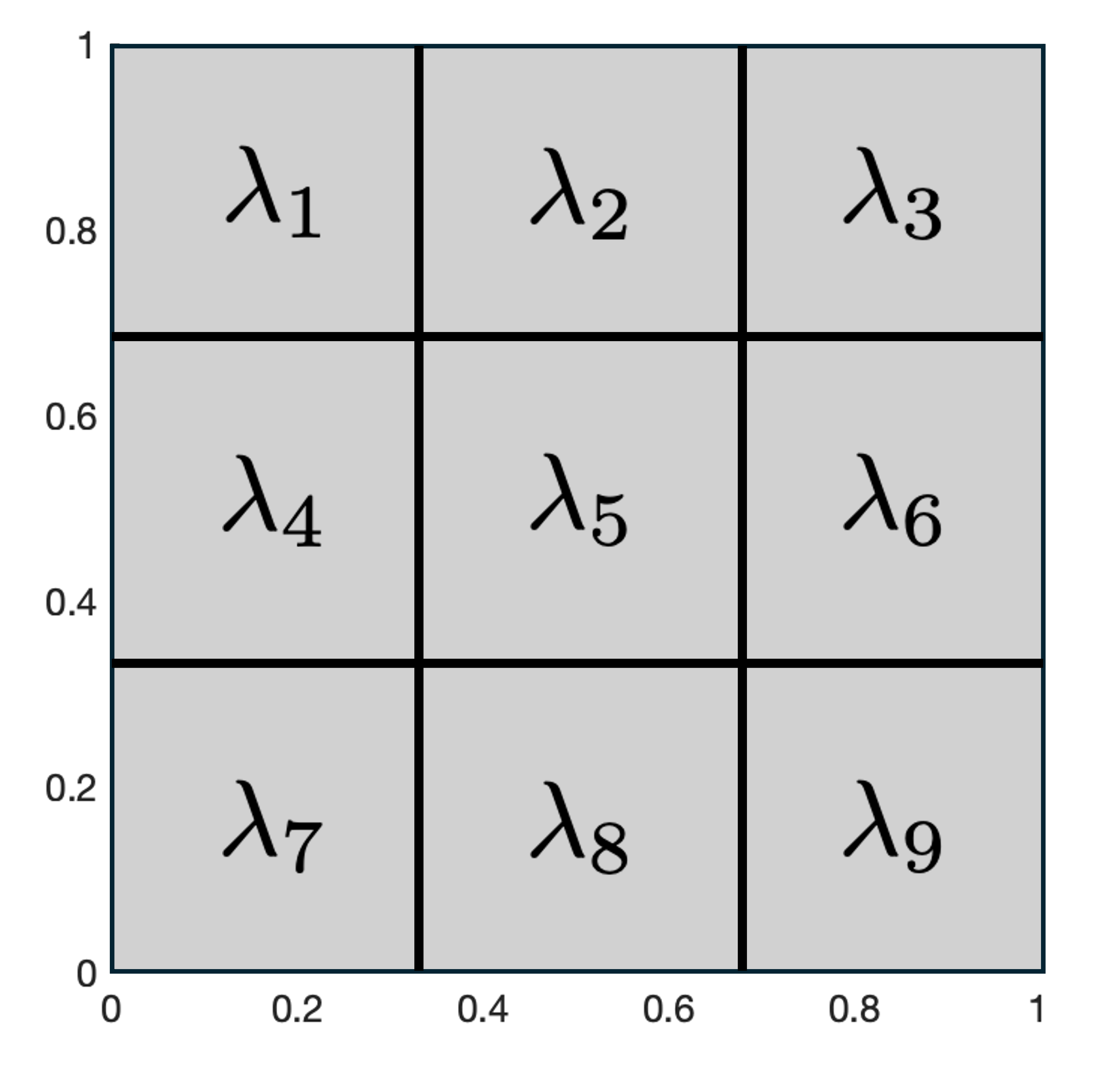}}
\caption{The domain $\Omega$ is constructed by welding together nine square plates, each of which has constant, but uncertain, thermal conductivity $\kappa$.}
\label{fig:heatplate_omega}
\end{figure}
Given a sample of thermal conductivity values for each $\lambda_i$, we approximated solutions to the PDE in Eq.~\eqref{eq:cbayes_heatrod} using MrHyDE with a $100 \times 100$ uniform rectangular mesh, piecewise linear finite elements, and a second-order implicit midpoint time stepping using 40 uniform time steps.

As in the previous section, we assume it is possible to place a thermometer and measure the solution at any location within the two-dimensional domain.
However, we now assume that we have nine thermometers to place.
Computationally, we used the the nodes of the mesh as a means of discretizing and indexing the design space for the sake of simplicity.
However, even with just $101^2 = 10,201$ potential observation locations for any given thermometer, this results in a design space of $9$-dimensional observable maps consisting of 
\[
{{10201}\choose{9}} = O(10^{30})
\] 
unique possible QoI maps.
Iterating over all possible maps to determine the OED is thus an infeasible computation.
We therefore utilized the greedy algorithm summarized in Algorithm~\ref{alg:greedy} to sequentially select measurement locations that are expected to add the most new geometric information to a given configuration of measurements.  

The space of all potential scalar-valued measurements defining $\Qspace^{(1)}$ in the computational experiments used the mesh points $(x_0^{(k)}, x_1^{(k)})\in\Omega$ to index $Q_k\in\Qspace^{(1)}$, i.e., 
\[
    Q_k=u\big((x_0^{(k)}, x_1^{(k)}); t_f\big),
\] 
for $1\leq k\leq 10201$ where $(x_0^{(k)},x_1^{(k)})$ denotes the potential location for a thermometer in the spatial domain. 
These $Q_k$ denote all of the possible first components of the final vector-valued QoI maps constructed by Algorithm~\ref{alg:greedy}. 
To estimate $ESE$, we generated 1,000 samples, 
\[
    \{\lambda^{(i)}\}_{i=1}^{1000},
\] 
in $\pspace$ and approximated
\[
    \big\{\{J_{\lambda^{(i)}, Q_{k}}\}_{i=1}^{1000}\big\}_{k=1}^{10201}
\] 
using a finite difference method with a perturbation of 1.0E-5.
In Fig.~\ref{fig:designspace}, we show $ESK^{-1}$ over the indexing set $\Omega$ for this current design space $\Qspace^{(1)}$.  
The optimal design, $Q_\text{opt}\in\Qspace^{(1)}$, is located at the center of the domain, near the location of the source.  

\begin{figure}[htbp]
    \centering
    {\includegraphics[width=0.45\textwidth]{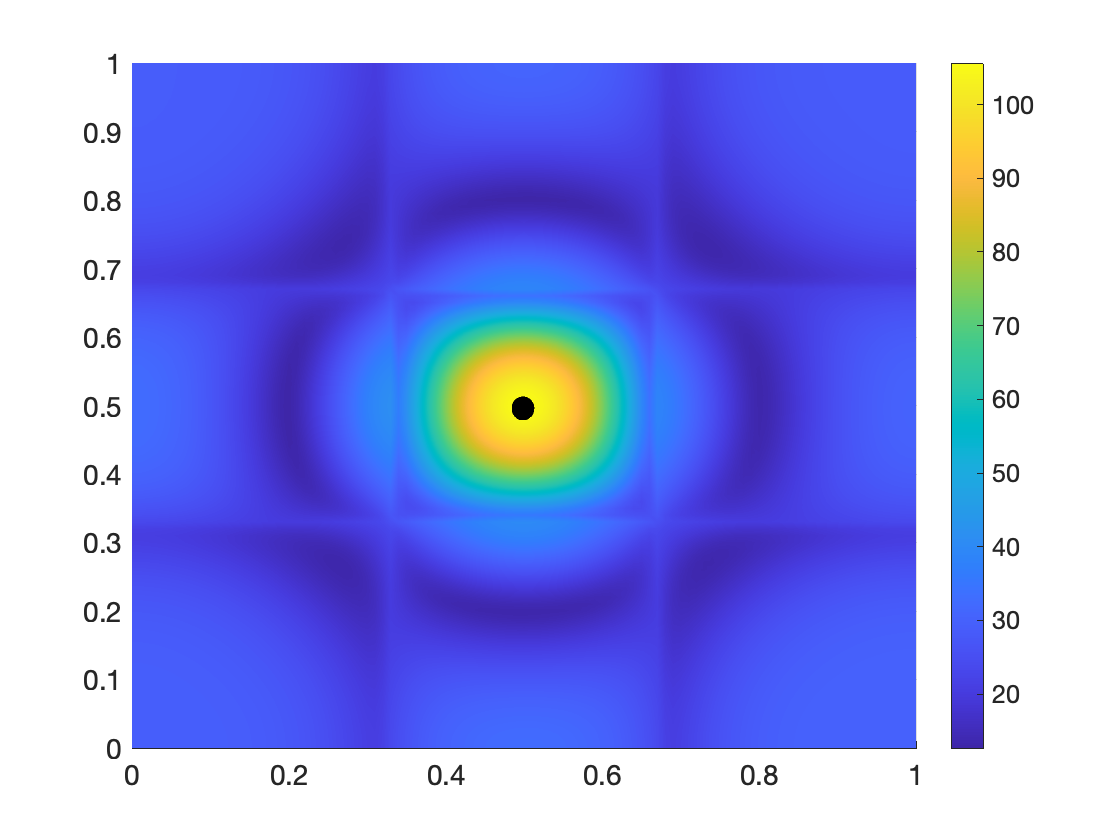}}
\caption{$ESK^{-1}$ over the indexing set for the design space $\Omega$.  
The location of the OED is shown in black, which is near the center of the domain where the source is primarily concentrated.}
\label{fig:designspace}
\end{figure}

Next, we consider $\Qspace^{(2)}$, the selection of the second contact thermometer when the first component is fixed at $Q_\text{opt}\in\Qspace^{(1)}$.  
We used the same discretization as above to compute Monte Carlo estimates of $ESK^{-1}$ for each $Q\in\Qspace^{(2)}$ and illustrate $ESK^{-1}$ over the indexing set $\Omega$ in the top left of Fig.~\ref{fig:designspace_d}.  
Through this iteration, the algorithm has defined a sequential approximation to the OED as a pair of sensors, one defined by $Q_\text{opt}\in\Qspace^{(1)}$ (now shown in white for contrast), and the other located in the upper right of the domain (shown in black).
There are in fact four locations (the four corners of the domain) for the second sensor placement that would produce a pair of sensors and subsequent QoI map with similar $ESK^{-1}$ values.
This is expected due to the symmetry of the problem and the precise choice of the four depends slightly on the sampling error.  
For the sake of space, we show without comment the iterative OED solutions to $\Qspace^{(d)}$ for $d=5, 7$, and $9$, in the top right, bottom left, and bottom right, respectively, of Fig.~\ref{fig:designspace_d}, where the previous measurement locations are always shown in white, but the newest sensor location is always shown in black.  

\begin{figure}[htbp]
    \centering
    {\includegraphics[width=.45\textwidth]{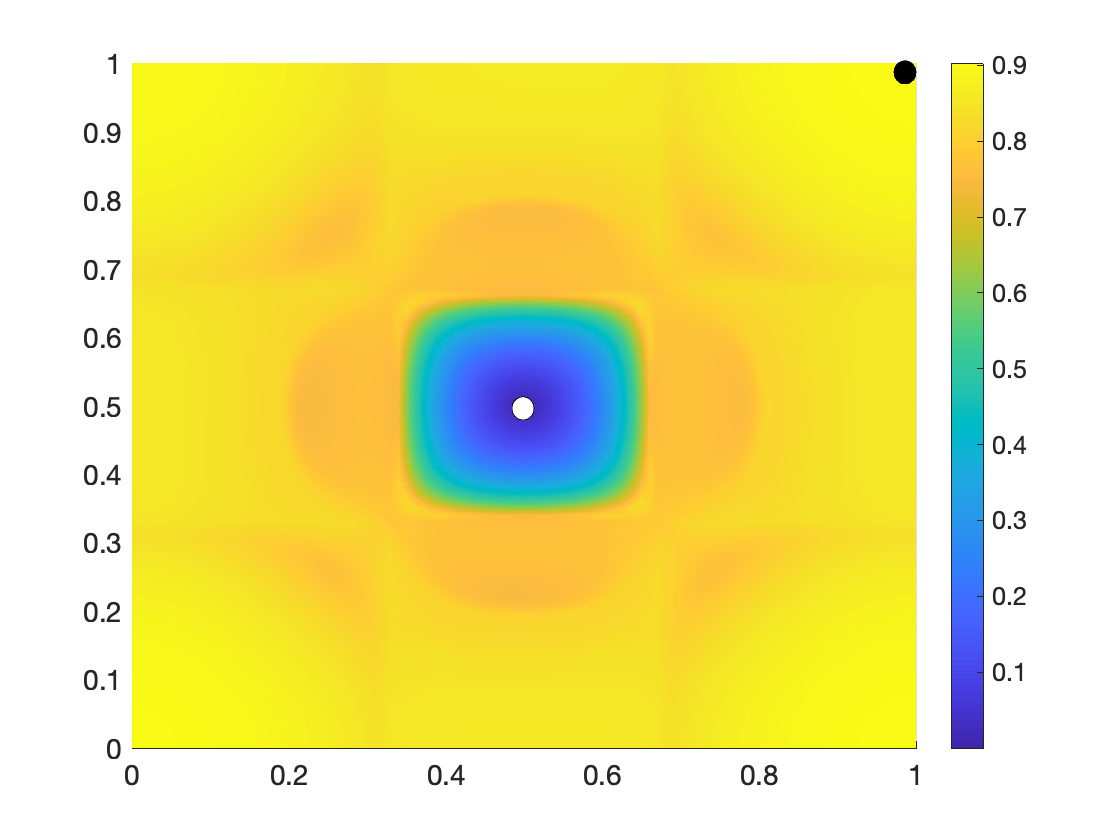}}
    {\includegraphics[width=.45\textwidth]{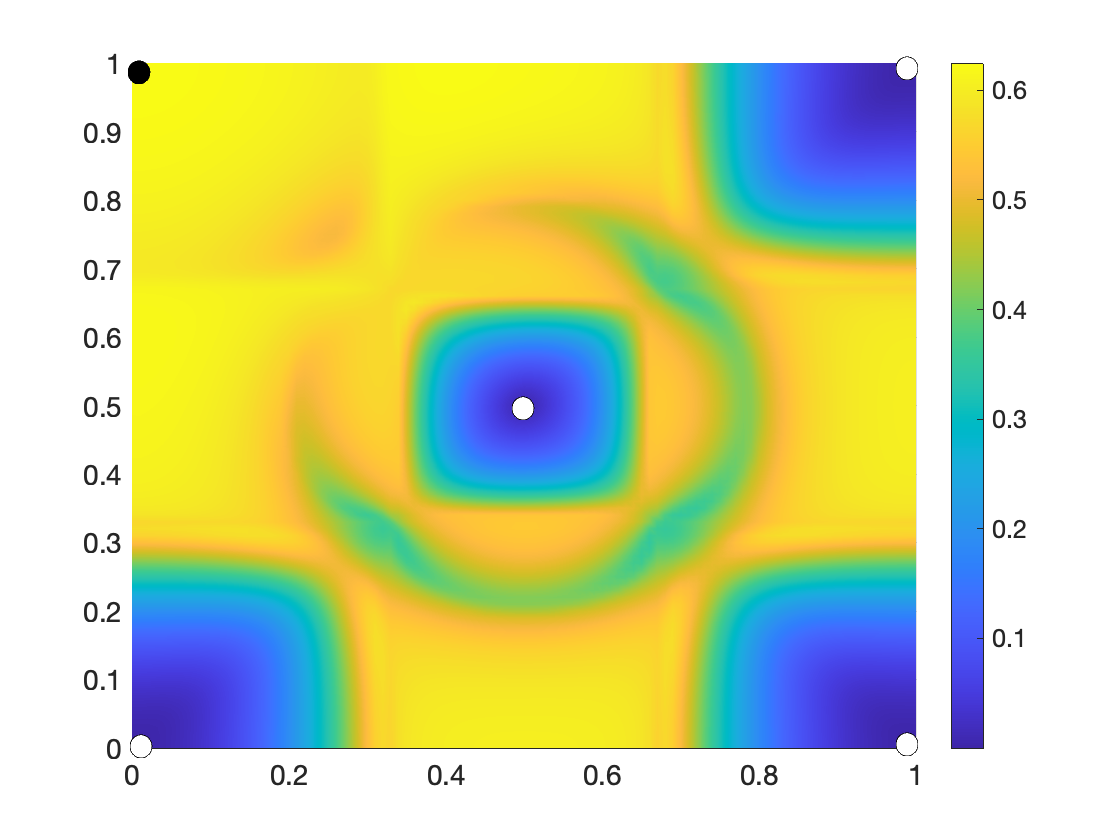}}
    {\includegraphics[width=.45\textwidth]{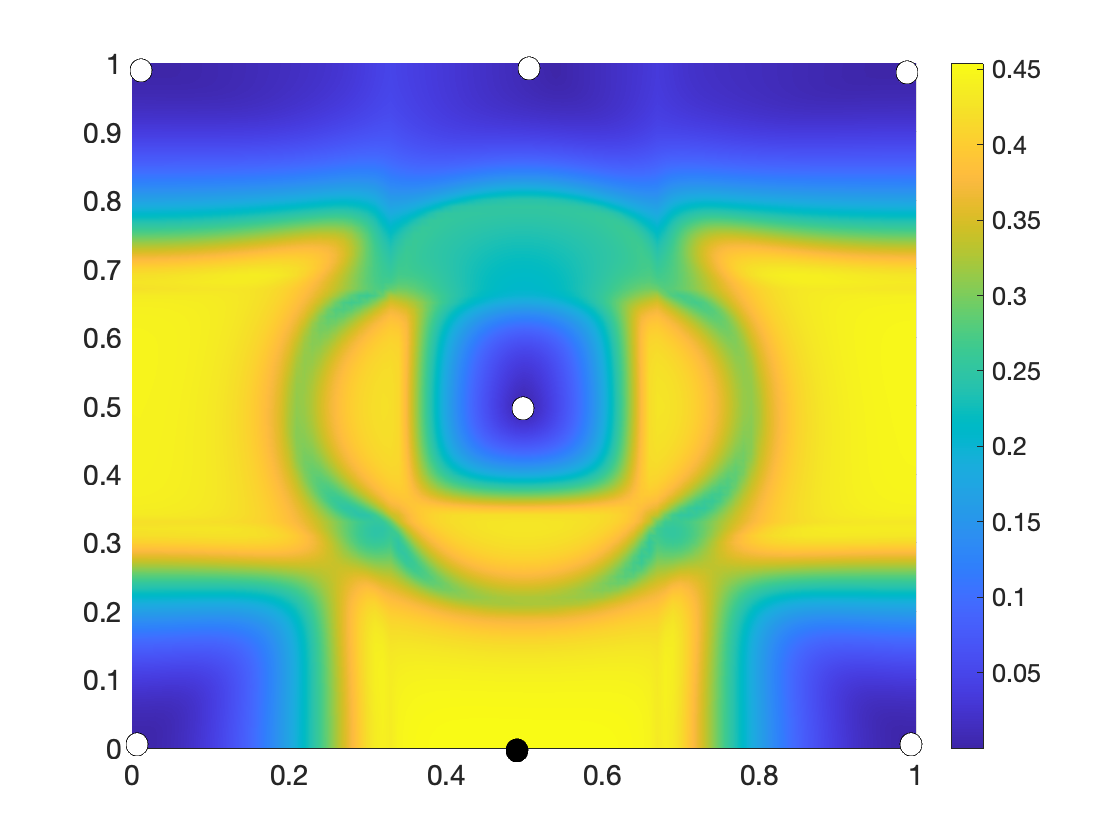}}
    {\includegraphics[width=.45\textwidth]{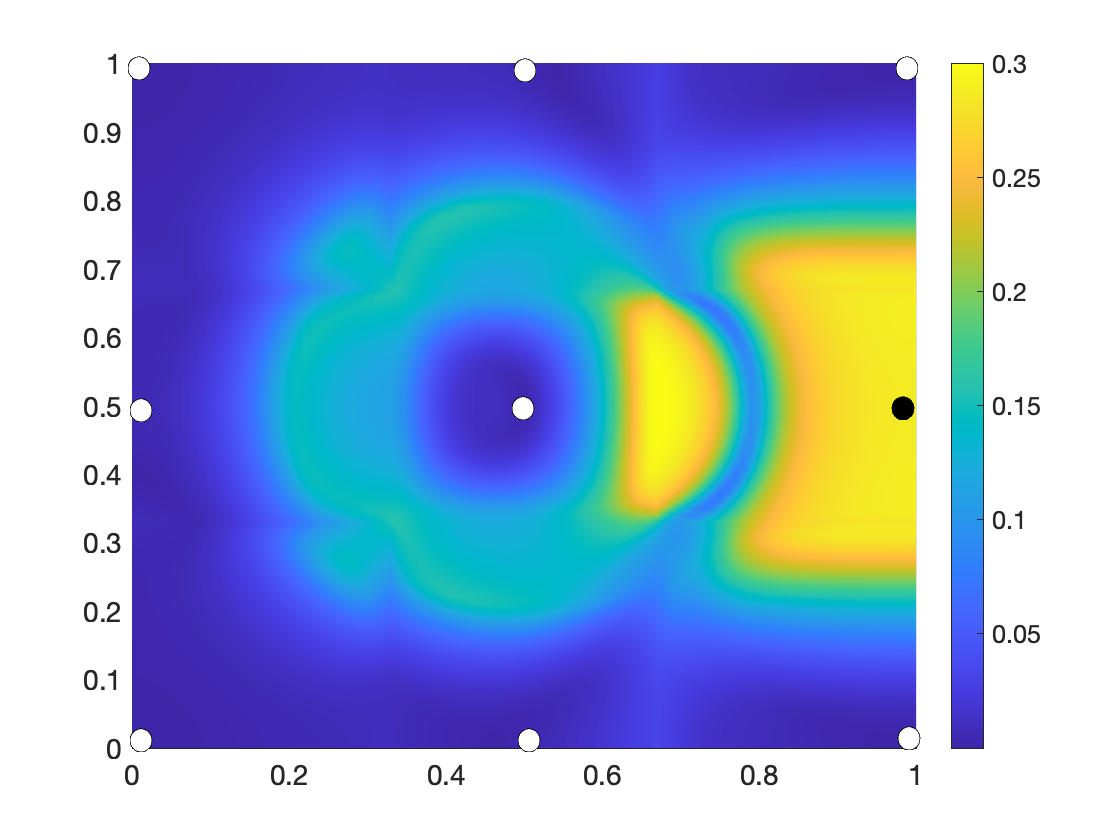}}
\caption{$ESK^{-1}$ for $\Qspace^{(d)}$ for $d=2$ (top left), $d=5$ (top right), $d=7$ (bottom left), and $d=9$ (bottom right) with the OED for $\Qspace^{(d-1)}$ shown as black dots and the OED for $\Qspace^{(d)}$ is shown as a white dot as defined by Algorithm~\ref{alg:greedy}.} 
\label{fig:designspace_d}
\end{figure}

The final configuration of sensors is both intuitive and interpretable while simultaneously providing a proof-of-concept of the greedy algorithm.
Specifically, intuition suggests each plate should receive a single sensor.
Furthermore, it is natural to assume that the sensor locations should be spaced as far apart from each other as possible to ensure that each associated component of the resulting QoI map is itself associated with data that are most influenced by the material properties of the particular plate in which the associated sensor is positioned.

\section{Conclusions}\label{sec:Conclusions}

In this work, we develop a general computational framework using singular values of sampled Jacobians of potential observable maps to define two optimal design criteria.
These design criteria are defined by distinct geometric properties of the maps referred to as skewness and scaling effects that describe how the maps invert sets of data. 
The use of singular values to compute these geometric properties was justified by theoretical results. 
We compared the impact of these design criteria on solutions to a stochastic inverse problem for a motivating problem describing the heating of a thin metal rod.
A greedy algorithm that combines the scaling and skewness effects was applied to a higher-dimensional example involving the welding together of nine different metal plates and produced an intuitive and interpretable result.

\section{Acknowledgments}

T.~Butler's work is supported by the National Science Foundation under Grant No.~DMS-2208460.
T.~Butler's work is also supported by NSF IR/D program while working at National Science Foundation. 
However, any opinion, finding, and conclusions or recommendations expressed in this material are those of the author and do not necessarily reflect the views of the National Science Foundation.

This paper describes objective technical results and analysis. Any subjective views or opinions that might be expressed in the paper do not necessarily represent the views of the U.S. Department of Energy or the United States Government.
Sandia National Laboratories is a multimission laboratory managed and operated by National Technology and
Engineering Solutions of Sandia, LLC., a wholly owned subsidiary of Honeywell International, Inc.,
for the U.S. Department of Energy’s National Nuclear Security Administration under contract DE-NA-0003525.

This material is based upon work supported by the U.S. Department of Energy, Office of Science,
Office of Advanced Scientific Computing Research, under contract 24-028431 and the Early Career Research Program.

\section{Data Availability Statement}

Details related to the software and discretization of the model problems utilized for generating the simulated data in the numerical results are provided in Section~\ref{S:Example}.
The scripts and input files needed to generate the results using MrHyDE can be found at:

\noindent https://github.com/sandialabs/MrHyDE/tree/main/scripts/DCI/OED-SVD


\bibliographystyle{plain}
\bibliography{ReferencesBib, oed4p}

\appendix

\section{Proof of Theorem~\ref{thm:local_inverse}}\label{app:dimension_assump_theorem}
\begin{proof}
We first prove the result for the linear case.
Assuming $Q$ is linear, then $Q(\lambda)=J_Q\lambda$ for a fixed $J_Q\in\mathbb{R}^{m\times n}$. 
If $m>n$, then there are at most $n$ rows of $J_Q$ that are linearly independent. 
If $m\leq n$, then there are most $m$ rows of $J_Q$ that are linearly independent.
In either case, let $r \leq n$ denote the number of rows of $J_Q$ that are linearly independent.
Basic results from linear algebra imply that $\dspace=Q(\pspace)\subset\mathbb{R}^d$ is defined by the column space of $J_Q$ and is given by a vector subspace of dimension $m$.
Thus, $\dspace$ is defined by an $r$-dimensional hyperplane embedded in $\mathbb{R}^m$, and $\dborel$ is a Borel $\sigma$-algebra that can be constructed by $r$-dimensional generalized rectangles defined with respect to an orthogonal set of $r$ vectors in $\mathbb{R}^m$. 

Let $\widehat{J}_Q$ denote any submatrix of $J_Q$ defined by choosing any $m$ linearly independent rows of $J_Q$, and denote by $\widehat{Q}(\lambda)$ the map into $\mathbb{R}^r$ defined by $\widehat{J}_Q\lambda$.
We similarly denote $\widehat{\dspace} = \widehat{Q}(\pspace)\subset\mathbb{R}^r$ and $\widehat{\dborel}$ as the Borel $\sigma$-algebra of $r$-dimensional sets in an $r$-dimensional space. 
Let $P$ denote a $r\times m$ projection matrix from $\dspace$ to $\widehat{\dspace}$. 
We abuse notation and let $PE$ denote the action of the matrix $P$ on all vectors in $E\in\dborel$, i.e., $PE\in\widehat{\dborel}$ is the projection of the event $E$ onto $\widehat{\dspace}$. 
Thus, for any $E\in \dborel$, there is a unique $\widehat{E}=PE\in\widehat{\dborel}$ representing the projection of $E$ onto $\widehat{\dspace}$.
The implication is that for any $E\in\dborel$, $Q^{-1}(E)=\widehat{Q}^{-1}(\widehat{E})\in\pborel.$

Now assume $Q$ is nonlinear but differentiable.
It follows that $Q$ is measurable so that $Q^{-1}(E)\in\pborel$.
Let $\epsilon>0$.
A standard result in measure theory (e.g., see Theorem~2.40 in \cite{Folland_book}) is that for any $A\in\pborel$, there exists a finite number, $K$, of generalized rectangles, $\set{A_k}_{1\leq k\leq K}$, such that if $\Delta$ denotes the symmetric difference between $A$ and the union of the $A_k$, then $\pmeas\left(\Delta\right)<\epsilon$.
Let $A=Q^{-1}(E)$ and apply such a result. 
For each $1\leq k\leq K$, denote a choice of $\lambda\in A_k$ as $\lambda^{(k)}$ and let $\tilde{Q}$ denote an $m$-dimensional piecewise linear map defined by
\begin{equation*}
    \tilde{Q}(\lambda) = \sum_{k=1}^K \left(Q(\lambda^{(k)})+J_{Q}(\lambda^{(k)})\lambda\right)\mathbf{I}_{A_k}(\lambda),
\end{equation*} 
where $\mathbf{I}_{A_k}(\lambda)$ denotes the indicator function on $A_k$ for each $1\leq k\leq K$. 
By the constant rank assumption, we can define $\widehat{Q}$ as the $r$-dimensional piecewise linear map constructed from $\tilde{Q}$ by taking $r\times n$ submatrices of $J_{Q}(\lambda^{(k)})$ and the same $r$ rows of $Q(\lambda^{(k)})$ for each $1\leq k\leq K$. 
Finally, let 
\begin{equation*}
    \widehat{E}_k = \widehat{Q}(A_k), \text{  for each } 1\leq k\leq K.
\end{equation*}
The conclusion follows by the choice of $A_k$ and the properties of $\widehat{Q}$.
\end{proof}
\section{Proofs for Scaling and Skewness Results}\label{app:scaling-skewness-proofs}

Below, we prove Lemma~\ref{lem:prod_singvals}, which stated that the volume of an embedded parallelepiped defined implicitly by the rows of a full rank $d\times n$ matrix is given by the product of singular values of $J$.
\begin{proof}
The singular values of $J$ are equal to the singular values of $J^{\top}$.  Consider the reduced QR factorization of $J^{\top}$,
\begin{equation}
    J^{\top} = \tilde QR,
\end{equation}
where $\tilde Q$ is $n\times m$ and $R$ is $m\times m$.  
By the properties of the QR factorization, we know the singular values of $R$ are the same as the singular values of $J^{\top}$.  Let $\mbf{x}\in\mathbb{R}^m$, then
\begin{equation}
	||\tilde Q\mbf{x}||^2=(\tilde Q\mbf{x})^{\top}(\tilde Q\mbf{x})=\mbf{x}^{\top}\tilde Q^{\top}\tilde Q\mbf{x}=\mbf{x}^{\top}\mbf{x}=||\mbf{x}||^2,
\end{equation}
so $\tilde Q$ is an isometry.
This implies the Lebesgue measure of the parallelepiped defined by the rows of $R$ is equal to the Lebesgue measure of the parallelepiped defined by the columns of $J^{\top}$, or the rows of $J$,
\begin{equation}
    \mu_d(Pa(J)) = \mu_d(Pa(R)) = \prod_{k=1}^m \gamma_k = \prod_{k=1}^m \sigma_k,
\end{equation}
where $\set{\gamma_k}_{1\leq k\leq m}$ are the singular values of $R$ and $\set{\sigma_k}_{1\leq k\leq m}$ are the singular values of $J$.
\end{proof}

Below, we prove Corollary~\ref{cor:cross_section} stating that the volume of an embedded parallelepiped defined by the cross-section of the pre-image of a $d$-dimensional unit cube under $d\times n$ full-rank matrix $J$ is given by the inverse of the product of the singular values of $J$. 
\begin{proof}
Consider the pseudo-inverse of $J$
$$
	J^\dagger = J^{\top}(JJ^{\top})^{-1}.
$$
From this equation, it is clear that the column space of $J^\dagger$ is equal to the row space of $J$.  The row space of $J$ defines a subspace orthogonal to the $n$-dimensional cylinder that is the pre-image of a unit cube under $J$.  Therefore, the column space of $J^\dagger$ is orthogonal to the pre-image cylinder and $Pa((J^\dagger)^{\top})$ is an $m$-dimensional parallelepiped defining an orthogonal cross-section of this cylinder.

From basic results in linear algebra, the singular values of $(J^\dagger)^{\top}$ are equal to those of $J^\dagger$.
Then, from properties of the pseudo-inverse, the singular values of $J^\dagger$ are the inverse of the singular values of $J$.  
Finally, from Lemma~\ref{lem:prod_singvals}, it follows that
\begin{equation}
    \mu_d(Pa((J^\dagger)^{\top})) = \left(\prod_{k=1}^{m} \, \sigma_{k}\right)^{-1},
\end{equation}
where $\set{\sigma_k}_{k=1}^m$ are the singular values of $J$.
\end{proof}

Below is the proof for Corollary~\ref{Cor:singular_values}.

\begin{proof}
From the definition of skewness and Theorem~\ref{thm:fundamental_decomp}, we have
\begin{equation}
    SK_Q(\lambda) = \max_{1\leq k\leq m} \, \frac{\|\mbf{j}_k\|}{\|\mbf{j}_k^{\perp}\|} = \max_{1\leq k\leq m} \, \frac{\|\mbf{j}_k\|\mu_{m-1}(Pa(J_{k}(\lambda)))}{\mu_m(Pa(J({\lambda})))}.
\end{equation}
Applying Lemma~\ref{lem:prod_singvals} finishes the proof.
\end{proof}





\end{document}